\documentclass
[
    a4paper,
    DIV=11,
    abstract=true,
]
{scrartcl}

\usepackage
{
    amsfonts,
    amsmath,
    amssymb,
    amsthm,
    authblk,
    dsfont, 
    enumitem,
    graphicx,
    mathtools,
    nicefrac,
    tikz,
    xcolor,
    bbm,
    mathrsfs,
    setspace,
    tabularx
}

\usepackage[scr=boondoxupr]{mathalfa}

\usepackage[utf8]{inputenc}

\usepackage[pdffitwindow=false,
            plainpages=false,
            pdfpagelabels=true,
            pdfpagemode=UseOutlines,
            pdfpagelayout=SinglePage,
            bookmarks=false,
            colorlinks=true,
            hyperfootnotes=false,
            linkcolor=blue,
            urlcolor=blue!30!black,
            citecolor=green!50!black]{hyperref}

\usepackage[bf,normal]{caption}

\usetikzlibrary{arrows, arrows.meta, shapes, fit, automata, positioning}

\graphicspath{{pics/}}

\DeclareMathAlphabet{\mathpzc}{OT1}{pzc}{m}{it}

\newcommand{\subfiguretitle}[1]{{\scriptsize{#1}} \\}
\newcommand{\R}{\mathbb{R}}                                      
\newcommand{\innerprod}[2]{\left\langle #1,\, #2 \right\rangle}  
\newcommand{\ts}{\hspace*{0.1em}}                                
\newcommand{\mc}[2][]{\mathpzc{#2}{\smash[t]{\mathstrut}}_{#1}}  

\DeclareMathOperator{\diag}{diag}
\DeclareMathOperator{\var}{var}
\DeclareMathOperator{\cov}{cov}
\DeclareMathOperator{\corr}{corr}

\newtheorem{theorem}{Theorem}[section]
\newtheorem{corollary}[theorem]{Corollary}
\newtheorem{lemma}[theorem]{Lemma}
\newtheorem{proposition}[theorem]{Proposition}
\newtheorem{definition}[theorem]{Definition}
\theoremstyle{definition}
\newtheorem{example}[theorem]{Example}
\newtheorem{remark}[theorem]{Remark}
\newtheorem{algorithm}[theorem]{Algorithm}

\newcommand\Wtilde{\stackrel{\sim}{\smash{\mc{G}}\rule{0pt}{1ex}}}
\newcommand\Vtilde{\stackrel{\sim}{\smash{\mc{V}}\rule{0pt}{1ex}}}
\newcommand\Etilde{\stackrel{\sim}{\smash{\mc{E}}\rule{0pt}{1ex}}}

\newenvironment{examplewithend}{%
  \pushQED{\qed}%
  \example%
}{%
  \popQED\endexample%
}
\makeatother

\makeatletter
\renewcommand*\env@matrix[1][*\c@MaxMatrixCols c]{%
  \hskip -\arraycolsep
  \let\@ifnextchar\new@ifnextchar
  \array{#1}}
\makeatother

\mathtoolsset{centercolon} 

\allowdisplaybreaks

\DeclareCaptionLabelFormat{period}{#1~#2}
\begin{document}

\title{Clustering Time-Evolving Networks \\ Using the Spatio-Temporal Graph Laplacian}
\author[1, a]{Maia Trower}
\author[2]{Nata\v sa Djurdjevac Conrad}
\author[3]{Stefan Klus}
\affil[1]{Maxwell Institute for Mathematical Sciences, University of Edinburgh and Heriot--Watt University, EH8 9BT, UK}
\affil[2]{Zuse Institute Berlin, 14195 Berlin, Germany}
\affil[3]{School of Mathematical \& Computer Sciences, Heriot--Watt University, EH14 4AS, UK}
\affil[a]{Corresponding author: maia.trower@ed.ac.uk}
\date{}

\maketitle

\begin{abstract}
Time-evolving graphs arise frequently when modeling complex dynamical systems such as social networks, traffic flow, and biological processes. Developing techniques to identify and analyze communities in these time-varying graph structures is an important challenge. In this work, we generalize existing spectral clustering algorithms from static to dynamic graphs using \emph{canonical correlation analysis} (CCA) to capture the temporal evolution of clusters. Based on this extended canonical correlation framework, we define the spatio-temporal graph Laplacian and investigate its spectral properties. We connect these concepts to dynamical systems theory via transfer operators, and illustrate the advantages of our method on benchmark graphs by comparison with existing methods. We show that the spatio-temporal graph Laplacian allows for a clear interpretation of cluster structure evolution over time for directed and undirected graphs.
\end{abstract}

\begin{spacing}{1.0}
\footnotesize \bfseries Community detection on static graphs is a well-studied problem, but extending this to time-evolving graphs is known to have significant challenges. This paper draws on transfer operator theory and presents an extension of canonical correlation analysis in order to propose a novel algorithm for detecting evolving cluster structure in dynamic graphs. Our algorithm is naturally able to successfully cluster both directed and undirected graphs and we illustrate the performance on benchmarks exhibiting a range of behaviors, including merging and splitting of clusters and clusters that grow and shrink in size.
\end{spacing}

\section{Introduction}

Graphs represent interactions, relationships, and associations between entities in a wide range of applications, from modeling social networks to transportation and to predicting climate phenomena~\cite{Bedi_2016, Borgatti_2018, Ganin_2017, Donges_2009}. However, much of the existing research on graphs has focused on static graphs, and so many approaches to analysis on dynamic graphs remove the temporal aspect of the graph in order to simplify, which omits crucial information about the structure or relationships. For example, transportation networks may change significantly due to road closures, and accurately modeling traffic flow depends upon having this information. In this work, we propose a novel algorithm for community detection in time-evolving graphs, where an extension of \emph{canonical correlation analysis} is used to leverage the temporal information. Community detection, or clustering, in graphs is a well-established and fundamental concept that partitions graphs into groups of vertices. These groups can be defined by relatively large numbers of intra-group edges, and relatively small numbers of inter-group edges. In the case of directed graphs, groups can also be characterized by sets of vertices with edges connecting the group to another group of vertices. That is, the group of vertices are related to each other because they have common in-links or out-links. Community detection can uncover meaningful relationships between interconnected vertices and can reveal structure across the graph~\cite{Lambiotte_2022}.

Spectral clustering in particular has been well-studied for static graphs \cite{luxburg_2007, Lambiotte_2022}. Recent work has shown that spectral clustering is related to transfer operators defined on graphs~\cite{Djurdjevac12, klus_2022, klus_2024}, and that this transfer operator-based formulation is valid for undirected, directed, and also time-evolving graphs. Building upon previous work showing that the eigenvectors of graph Laplacians can be interpreted as eigenfunctions of associated transfer operators and then used to partition the vertices of a graph~\cite{klus_2024}, we extend this spectral clustering framework to time-evolving graphs by considering the sequence of adjacency matrix snapshots and adapting canonical correlation analysis to extract eigenvectors which contain information relating to cluster structure over the entire time period. These eigenvectors contain information about the cluster structure of a ``flattened" static graph associated with the time-evolving graph, and we show that the flattened graph encodes both spatial and temporal information from the original graph. Crucially, it also retains information on the directionality of edges, and we are able to detect changes in cluster structure. The method in~\cite{klus_2024} is able to identify sets of vertices that are coherent over the whole time interval, but cannot detect when clusters merge, split, or otherwise change. The method proposed in this work addresses this and shows the evolution of clusters.

Clustering time-evolving graphs is known to be a difficult problem. It is not clear how to define clusters when edges evolve over time, and it is also not obvious whether we should cluster static snapshots and then match together static clusters, or find clusters in a way that accounts for the temporal nature of the graph~\cite{rossetti_2018}. In general, the answer to these methodological questions are context-dependent, and there are existing methods in the literature that ascribe to many different perspectives. The simplest approach involves an aggregation of all edges between vertices at any time step and then clustering the resulting static graph using existing methods, but this approach neglects temporal information entirely~\cite{onnela_2007}. However, in general it is not desirable to remove all dynamic information from a time-evolving graph in order to study the communities contained within the graph. To retain temporal information, the main approaches can broadly be categorized as a) identification of communities at the time of each snapshot of the graph, or b) identification of the life cycle of communities throughout the evolution of the graph~\cite{aynaud_2013}.

Approach a) includes \textit{two-stage} methods, which typically detect clusters at time $ t $ and then infer relationships between clusterings found at different time steps. Examples of this can be found in~\cite{rosvall_2010, palla_2007, Lorenz-Spreen_2018}. These methods are all subject to the same major drawback; smooth transitions between clusterings are difficult to achieve, and parameter choices or methods to solve the matching (of which many exist, see~\cite{cazabet_2021}) have a large impact on the output. Approaches falling in category b) may construct a \textit{coupling graph} with edges representing temporal relationships between vertices. Examples of this can be found in~\cite{jdidia_2007, mucha_2010}, and we mention in detail the approach considered in~\cite{sole_ribalta_2013}. This method couples time layers together and defines a \emph{supra-Laplacian} which can be used to identify communities via spectral clustering. Such methods typically rely on parameters that must be tuned, and in particular there is no standard method for choosing the coupling strength between time steps for the supra-Laplacian, so the resulting communities are heavily dependent on the parameters chosen.

We also note the connection between clustering time-evolving networks and notions of metastability and coherence in dynamical systems. Metastable sets \cite{Davies82a, Davies82b, Bovier16} are sets such that trajectories in dynamical systems will remain within the set with a high probability, and trajectories outside this set will remain outside with a high probability. This concept can be applied, for example, to molecular dynamics to detect conformations of molecules~\cite{SKH23}, and we can interpret clusters in undirected graphs as metastable states of random walkers on the graph \cite{klus_2024_review}. This notion can be extended to time-dependent metastability, known as \emph{coherence}, with the following idea: Coherent sets \cite{Froyland_2010, Froyland_2013, Gehle_2017} in dynamical systems are time-dependent sets whose boundaries in space-time are crossed by trajectories with low probability \cite{fackeldey_2019}. Such sets play an important role in fluid dynamics, where they represent, for instance, slowly mixing eddies or gyres in the ocean \cite{FSvS14}. This is an intuitive expansion of metastability, where the boundaries of metastable sets are defined only in space. We explore the parallels between detecting coherence in dynamical systems and our spatio-temporal Laplacian approach for graph clustering in later sections of this work. Further, we consider the potential for communities of vertices which can be regarded as clusters because they have common in- or out-links. These kinds of clusters are key in many applications, including in neuroscience, biology, and social networks. We also refer to \cite{Hlinka_2017, Zhou_2015} for examples of how accounting for directionality in climate systems can  lead to important discoveries about the system; so-called \emph{teleconnections} between geographically remote regions can be used to explain or predict climate and weather phenomena. Therefore, detecting directed clusters in climate networks can be interpreted as identifying globally disparate regions all contributing to the same pattern or phenomena.

Canonical correlation analysis is a multivariate statistical technique that is used for measuring the linear relationships between two multidimensional variables \cite{hotelling_1936}, which seeks to identify linear transformations of variables that maximizes the correlation between these transformations. It has been shown in~\cite{Klus_2019} that CCA is related to the forward--backward dynamics of stochastic differential equations and can be used for the detection of coherent sets in time series data by considering only the first and last time steps. It was shown in~\cite{klus_2024} that CCA is also related to the forward--backward dynamics of a graph and that maximizing the correlation via CCA is equivalent to finding the eigenfunctions of the forward--backward operator which can then be used to cluster directed graphs. This motivates our approach, where we aim to extend CCA by maximizing the correlation across time steps (which we now refer to as \emph{views}) and then using the resulting eigenfunctions in order to maximize coherence over the time interval.

In~\cite{cristianini_2004}, a natural extension of CCA, called \emph{multiview CCA} (mCCA), is proposed which maximizes coherence across all views. In this work we adapt mCCA to consider only the maximization of correlations between adjacent time views. Using this formulation of mCCA, we can identify relationships between variables at intermediate times, as opposed to only considering the start points and end points of the time interval. This is key as we specifically aim to detect changes in coherence over time, and not to simply identify sets that are coherent across the whole interval.  This is particularly useful in cases where we are interested in data that evolve in time, where the relationship dynamics are changing throughout the time period. For example, if we consider a transportation network, road closures could be represented by the removal of edges, which would reappear once the road has reopened. The closing and reopening of the same road would not be detected if we applied CCA on the first and last time steps only.

Using mCCA, we propose a \textit{spatio-temporal graph Laplacian} and then apply a spectral clustering algorithm to detect coherent sets in time-evolving graphs and characterize the evolution of these sets. Unlike existing methods, the spatio-temporal graph Laplacian does not require parameter tuning. Also, the eigenvalues of the spatio-temporal graph Laplacian are always real-valued so our method is also capable of clustering directed time-evolving graphs. The main contributions of this work are:
\begin{itemize}[leftmargin=4.5ex, itemsep=0ex]
\item We extend definitions of transfer operators, covariance operators and related concepts to time-evolving graphs, and study their properties.
\item We introduce a variant of canonical correlation analysis capable of maximizing correlations across multiple time slices, and provide theoretical results relating to this method.
\item We define the spatio-temporal graph Laplacian based on this extended canonical correlation analysis and illustrate how it can be used to detect communities in time-evolving graphs.
\item We analyze the performance of our approach using different types of benchmark graphs with clusters that exhibit splitting and merging behavior and clusters that grow or shrink in size. We compare our results with other spectral clustering approaches.
\end{itemize}

In Section~\ref{sec:background}, we present relevant definitions for graphs, transfer operators, and other mathematical concepts. We proceed by formulating an extension to canonical correlation analysis in Section~\ref{sec:sc_teg}, where we also introduce and analyze the spatio-temporal graph Laplacian. The interpretation of this Laplacian is explored and a novel spectral clustering algorithm is presented. Finally, we evaluate this algorithm in Section~\ref{sec:num_results} by comparing our method to existing methods found in the literature using constructed benchmark graphs. We conclude with a discussion and remark on open questions in Section~\ref{sec:discussion}.

\section{Background} \label{sec:background}

In this section we outline some basic definitions of graphs and briefly describe relevant transfer operators. For more details, see~\cite{klus_2024}.

\subsection{Graphs}

We begin by introducing static and time-evolving graphs and some key related concepts. In this work we largely refer to undirected graphs but we remark throughout how the method can also be applied to directed graphs.

\begin{definition}[Weighted graph]
Let $ \mc{V} = \{\mc[1]{v}, \dots, \mc[n]{v}\} $ be the set vertices and $ \mc{E} \subseteq \mc{V} \times \mc{V} $ the set of edges. Then $ \mc{G} = (\mc{V}, \mc{E}) $ is the graph given by these edges and vertices. The graph $ \mc{G} $ can also be associated with a weighted adjacency matrix $ W \in \mathbb{R}^{n \times n} $, where
\begin{align*}
    W_{ij} = \begin{cases}
        w(\mc[i]{v}, \mc[j]{v}), & \text{if } (\mc[i]{v}, \mc[j]{v}) \in \mc{E}, \\
    0, & \text{otherwise},
    \end{cases}
\end{align*}
and $ w(\mc[i]{v}, \mc[j]{v}) > 0 $ is the weight of the edge $ (\mc[i]{v}, \mc[j]{v}) $. Note that in the special case where $ w(\mc[i]{v}, \mc[j]{v}) = w(\mc[j]{v}, \mc[i]{v}) $ for all $ i, j $, the graph is \emph{undirected}.
\end{definition}

\begin{definition}[Time-evolving graph]
A time-evolving graph $ \mc{G}^{(M)} = \{\mc[1]{G}, \dots, \mc[M]{G}\} $ is given by a set of graphs $ \mc[t]{G} = (\mc{V}, \mc{E_t}) $ defined at each time view $ t \in \{1, \dots, M\} $. As above, each snapshot $ \mc[t]{G} $ has an associated weighted adjacency matrix $ W_t $.
\end{definition}

Note here that the set of vertices $ \mc{V} $ remains unchanged for all time steps (and therefore the dimension of $ W_t $ is constant), and vertex labels are preserved throughout. In the literature, time-evolving graphs are referred to by many other names, including time-varying graphs~\cite{casteigts_2012}, dynamic networks~\cite{kuhn_2011}, and temporal networks~\cite{holme_2012}. In this work we refer to them as time-evolving graphs, and we call graphs that are not evolving in time static graphs. In general, time-evolving graphs can have vertices disappearing and reappearing from one time step to the next (i.e.,\ $ \mc{V} $ is also time-dependent). However, an analysis of such graphs is beyond the scope of this work.

\begin{definition}[Transition matrix]
For a static graph $ \mc{G} = (\mc{V}, \mc{E}) $ with $ n $ vertices, the transition matrix $ S \in \R^{n \times n} $ is defined as $ S = D_{\mathscr{o}}^{-1} W $, where
\begin{equation*}
    D_{\mathscr{o}} = \diag(\mathscr{o}(\mc[1]{v}), \dots, \mathscr{o}(\mc[n]{v})) \quad \text{and} \quad  \mathscr{o}(\mc[i]{v}) = \sum_{j=1}^n W_{ij}.
\end{equation*}
This transition matrix represents the probability that a random walker will move from $ \mc[i]{v} $ to $ \mc[j]{v} $ in one step. For time-evolving graphs, we can define the transition matrix at each view. The graph $ \mc{G}^{(M)} = \{\mc[1]{G}, \dots, \mc[M]{G}\} $ has associated transition matrices $ S^{(M)} = \{S_1, \dots, S_M\} $.
\end{definition}

\begin{definition}[Random walk on static graphs]
For a static graph $ \mc{G} = (\mc{V}, \mc{E}) $, a random walk is a discrete stochastic process starting in a vertex $ \mc[i]{v} $. At each time step the walk moves to another vertex $ \mc[j]{v} $ with probability $ S_{ij} $, where $ S $ is the transition matrix of $ \mc{G} $.
\end{definition}

\begin{definition}[Random walk on time-evolving graphs] \label{def: terw}
Let $ \mc{G}^{(M)} = \{\mc[1]{G}, \dots, \mc[M]{G} \} $ be a time-evolving graph with associated transition matrices $ S^{(M)} = \{S_1, \dots, S_M\} $. Then a time-evolving random walk on $ \mc{G}^{(M)} $ is a discrete stochastic process starting at $ \mc[i]{v} $ at $ t = 1 $ and moving to $ \mc[j]{v} $ at time $ t = 2 $ with probability $ (S_1)_{ij} $. The process then moves from $ \mc[j]{v} $ at $ t = 2 $ to $ \mc[k]{v} $ at $ t = 3 $ with probability $ (S_2)_{jk} $, and so on. That is, the transition probabilities between states of the random walk at time $ t $ are determined by the transition matrix of the time-evolving graph at time $ t $.
\end{definition}

\begin{examplewithend} \label{example1}

\begin{figure}
    \centering
    \vspace*{1ex}
    \resizebox{0.45\textwidth}{!}{%
    \begin{tikzpicture}[
            >= stealth, 
            semithick 
        ]
        \tikzstyle{every state}=[
            draw=black,
            thick,
            fill=white,
            inner sep=0pt,
            text width=8mm,
            align=center,
            scale=0.6,
            font = \Large
        ]
        \tikzset{every loop/.style={}} 

        \node[state] (v1) {$\mc[1]{v}$};
        \node[state] (v2) [right=0.5cm of v1] {$\mc[2]{v}$};
        \node[state] (v3) [right=0.5cm of v2] {$\mc[3]{v}$};
        \node[state] (v4) [right=0.5cm of v3] {$\mc[4]{v}$};
        \node[state] (v5) [right=0.5cm of v4] {$\mc[5]{v}$};
        \node[state] (v6) [right=0.5cm of v5] {$\mc[6]{v}$};

        \path (v1) edge [loop above, min distance=5mm, in=60, out=120] node {} (v1);
        \path (v2) edge [loop above, min distance=5mm, in=60, out=120] node {} (v2);
        \path (v3) edge [loop above, min distance=5mm, in=60, out=120] node {} (v3);
        \path (v4) edge [loop above, min distance=5mm, in=60, out=120] node {} (v4);
        \path (v5) edge [loop above, min distance=5mm, in=60, out=120] node {} (v5);
        \path (v6) edge [loop above, min distance=5mm, in=60, out=120] node {} (v6);

        \path[-] (v1) edge node {} (v2);
        \path[-, dotted] (v2) edge node {} (v3);
        \path[-] (v3) edge node {} (v4);
        \path[-, dotted] (v4) edge node {} (v5);
        \path[-] (v5) edge node {} (v6);

        \node[state] (v7) [below=1cm of v1]{$\mc[1]{v}$};
        \node[state] (v8) [right=0.5cm of v7] {$\mc[2]{v}$};
        \node[state] (v9) [right=0.5cm of v8] {$\mc[3]{v}$};
        \node[state] (v10) [right=0.5cm of v9] {$\mc[4]{v}$};
        \node[state] (v11) [right=0.5cm of v10] {$\mc[5]{v}$};
        \node[state] (v12) [right=0.5cm of v11] {$\mc[6]{v}$};

        \path (v7) edge [loop above, min distance=5mm, in=60, out=120] node {} (v7);
        \path (v8) edge [loop above, min distance=5mm, in=60, out=120] node {} (v8);
        \path (v9) edge [loop above, min distance=5mm, in=60, out=120] node {} (v9);
        \path (v10) edge [loop above, min distance=5mm, in=60, out=120] node {} (v10);
        \path (v11) edge [loop above, min distance=5mm, in=60, out=120] node {} (v11);
        \path (v12) edge [loop above, min distance=5mm, in=60, out=120] node {} (v12);

        \path[-] (v7) edge node {} (v8);
        \path[-, dashed] (v8) edge node {} (v9);
        \path[-] (v9) edge node {} (v10);
        \path[-, dotted] (v10) edge node {} (v11);
        \path[-] (v11) edge node {} (v12);

        \node[state] (v13) [below=1cm of v7]{$\mc[1]{v}$};
        \node[state] (v14) [right=0.5cm of v13] {$\mc[2]{v}$};
        \node[state] (v15) [right=0.5cm of v14] {$\mc[3]{v}$};
        \node[state] (v16) [right=0.5cm of v15] {$\mc[4]{v}$};
        \node[state] (v17) [right=0.5cm of v16] {$\mc[5]{v}$};
        \node[state] (v18) [right=0.5cm of v17] {$\mc[6]{v}$};

        \path (v13) edge [loop above, min distance=5mm, in=60, out=120] node {} (v13);
        \path (v14) edge [loop above, min distance=5mm, in=60, out=120] node {} (v14);
        \path (v15) edge [loop above, min distance=5mm, in=60, out=120] node {} (v15);
        \path (v16) edge [loop above, min distance=5mm, in=60, out=120] node {} (v16);
        \path (v17) edge [loop above, min distance=5mm, in=60, out=120] node {} (v17);
        \path (v18) edge [loop above, min distance=5mm, in=60, out=120] node {} (v18);

        \path[-] (v13) edge node {} (v14);
        \path[-] (v14) edge node {} (v15);
        \path[-] (v15) edge node {} (v16);
        \path[-, dotted] (v16) edge node {} (v17);
        \path[-] (v17) edge node {} (v18);

        \node[state] (v19) [below=1cm of v13] {$\mc[1]{v}$};
        \node[state] (v20) [right=0.5cm of v19] {$\mc[2]{v}$};
        \node[state] (v21) [right=0.5cm of v20] {$\mc[3]{v}$};
        \node[state] (v22) [right=0.5cm of v21] {$\mc[4]{v}$};
        \node[state] (v23) [right=0.5cm of v22] {$\mc[5]{v}$};
        \node[state] (v24) [right=0.5cm of v23] {$\mc[6]{v}$};

        \path (v19) edge [loop above, min distance=5mm, in=60, out=120] node {} (v19);
        \path (v20) edge [loop above, min distance=5mm, in=60, out=120] node {} (v20);
        \path (v21) edge [loop above, min distance=5mm, in=60, out=120] node {} (v21);
        \path (v22) edge [loop above, min distance=5mm, in=60, out=120] node {} (v22);
        \path (v23) edge [loop above, min distance=5mm, in=60, out=120] node {} (v23);
        \path (v24) edge [loop above, min distance=5mm, in=60, out=120] node {} (v24);

        \path[-] (v19) edge node {} (v20);
        \path[-] (v20) edge node {} (v21);
        \path[-] (v21) edge node {} (v22);
        \path[-, dotted] (v22) edge node {} (v23);
        \path[-] (v23) edge node {} (v24);

        \foreach \Point in {(1.13, -0.196),(1.29, -1.494),(0.926, -3.149),(5.319, -4.721),(0.048, -0.027),(0.199, -1.541),(-0.074, -3.095),(2.381, -4.823),(1.083, -0.15),(1.145, -1.579),(1.061, -3.218),(0.977, -4.75),(1.245, 0.058),(1.234, -1.665),(1.109, -3.345),(1.134, -4.606),(-0.092, 0.009),(-0.137, -1.591),(2.293, -3.352),(0.091, -4.868),(1.14, -0.159),(1.096, -1.396),(0.897, -3.151),(1.112, -4.637),(0.941, 0.074),(1.227, -1.466),(0.956, -3.193),(3.375, -5.002),(1.212, -0.145),(1.059, -1.456),(1.224, -3.156),(3.399, -4.859),(-0.159, 0.112),(0.099, -1.747),(2.153, -3.056),(-0.009, -4.85),(1.222, 0.149),(1.125, -1.541),(3.376, -3.153),(1.259, -4.932),(1.172, -0.136),(1.085, -1.495),(1.255, -3.221),(0.989, -4.755),(-0.025, -0.071),(0.087, -1.556),(2.109, -3.398),(2.064, -4.799),(-0.003, -0.083),(0.095, -1.632),(0.042, -3.36),(2.194, -4.621),(1.256, -0.147),(0.898, -1.573),(1.23, -3.235),(3.177, -4.917),(1.044, -0.184),(1.251, -1.471),(1.007, -3.293),(3.248, -4.758),(1.044, -0.156),(1.278, -1.665),(3.474, -3.168),(3.241, -4.884),(1.089, -0.178),(0.97, -1.554),(1.072, -3.218),(3.179, -4.801),(1.23, -0.024),(1.049, -1.433),(3.278, -3.382),(3.302, -4.84),(-0.02, -0.079),(2.209, -1.406),(-0.065, -3.197),(2.345, -4.697),(1.196, 0.156),(1.078, -1.607),(3.387, -3.11),(1.216, -4.953),(0.099, 0.18),(0.209, -1.552),(2.359, -3.199),(2.404, -4.735),(1.058, -0.178),(0.941, -1.55),(1.106, -3.279),(3.234, -5.001),(1.131, -0.039),(1.035, -1.554),(3.248, -3.13),(1.186, -4.79),(0.195, -0.041),(0.051, -1.525),(2.338, -3.043),(2.186, -4.865),(0.177, 0.079),(0.097, -1.781),(0.205, -3.247),(1.996, -4.835),(0.118, -0.18),(-0.024, -1.521),(0.099, -3.216),(0.175, -4.679),(-0.008, -0.199),(-0.033, -1.429),(0.175, -3.245),(-0.073, -4.853),(1.013, -0.073),(1.177, -1.79),(1.164, -3.333),(0.956, -4.731),(1.191, 0.019),(0.93, -1.473),(0.957, -3.23),(1.032, -4.6),(-0.118, 0.133),(0.171, -1.502),(-0.045, -3.251),(0.1, -4.829),(0.085, -0.148),(-0.026, -1.538),(0.052, -3.303),(2.099, -4.984),(1.191, -0.145),(1.302, -1.635),(3.225, -3.216),(3.398, -4.789),(-0.167, -0.0),(-0.011, -1.482),(-0.044, -2.993),(0.071, -4.614),(1.181, -0.164),(1.04, -1.746),(1.105, -3.261),(1.107, -4.84),(-0.189, 0.087),(0.017, -1.417),(2.387, -3.233),(2.277, -4.662),(0.018, -0.044),(-0.08, -1.548),(2.189, -3.192),(0.09, -4.997),(-0.023, 0.025),(-0.063, -1.616),(0.035, -3.326),(0.11, -4.669),(1.143, 0.142),(0.997, -1.554),(1.086, -3.128),(1.063, -4.896),(1.274, 0.123),(1.031, -1.433),(0.957, -3.208),(1.256, -4.894),(1.267, 0.008),(1.186, -1.629),(0.915, -3.168),(0.959, -4.672),(-0.05, 0.053),(0.031, -1.566),(0.161, -3.26),(2.131, -4.724),(1.188, -0.025),(1.162, -1.528),(1.218, -3.33),(3.302, -4.884),(1.156, -0.194),(1.303, -1.684),(3.371, -3.291),(3.177, -4.935),(1.125, 0.129),(1.079, -1.779),(1.224, -3.129),(1.03, -4.75),(-0.111, 0.181),(-0.141, -1.591),(0.075, -2.995),(-0.001, -4.804),(0.923, -0.078),(1.226, -1.551),(3.393, -3.008),(1.0, -4.971),(0.075, 0.112),(0.184, -1.59),(0.016, -3.39),(4.429, -4.972),(1.293, -0.057),(1.118, -1.483),(1.18, -3.245),(0.974, -4.804),(1.069, -0.061),(1.17, -1.712),(3.395, -3.016),(3.245, -5.006),(-0.183, -0.04),(-0.151, -1.715),(2.24, -3.108),(0.05, -4.842),(-0.029, -0.038),(0.156, -1.593),(2.37, -3.264),(2.113, -4.82),(0.907, 0.064),(1.134, -1.498),(1.001, -3.272),(1.243, -4.773)}
            \draw[green,fill=green] \Point circle (0.045ex);
        \foreach \Point in {(1.982, -0.013),(0.185, -1.51),(-0.021, -3.136),(1.999, -4.73),(2.244, 0.199),(4.529, -1.713),(4.358, -3.214),(4.399, -4.967),(3.248, 0.203),(3.125, -1.732),(1.259, -3.116),(1.04, -4.826),(2.111, -0.085),(4.424, -1.708),(4.4, -3.357),(4.594, -4.702),(2.062, 0.058),(2.361, -1.718),(2.192, -3.001),(2.128, -4.913),(3.376, 0.075),(5.52, -1.589),(5.642, -3.296),(3.413, -4.841),(3.236, -0.003),(3.231, -1.4),(3.18, -3.222),(3.169, -4.734),(2.195, 0.153),(2.06, -1.503),(1.985, -3.236),(0.212, -4.844),(3.134, 0.02),(3.166, -1.448),(0.932, -3.186),(1.121, -4.726),(3.177, 0.182),(3.468, -1.506),(3.427, -3.165),(3.423, -4.967),(2.21, -0.135),(2.278, -1.672),(-0.094, -3.021),(0.056, -4.784),(3.423, 0.084),(3.235, -1.742),(1.14, -3.328),(1.19, -4.735),(2.322, -0.092),(2.124, -1.709),(2.101, -3.098),(2.09, -4.666),(3.128, -0.065),(3.18, -1.698),(3.334, -3.368),(1.11, -4.851),(2.17, -0.126),(2.078, -1.636),(4.585, -3.165),(4.368, -4.81),(2.368, -0.044),(2.146, -1.589),(2.208, -3.296),(2.128, -4.864),(3.113, -0.066),(3.147, -1.546),(1.316, -3.196),(3.132, -4.688),(3.219, -0.027),(3.192, -1.489),(3.256, -3.353),(3.362, -4.869),(2.212, -0.185),(2.234, -1.692),(2.38, -3.308),(2.251, -4.992),(3.18, -0.064),(3.092, -1.648),(3.433, -3.056),(3.344, -4.809),(3.419, 0.047),(3.15, -1.645),(1.202, -3.294),(3.319, -5.002),(2.389, 0.031),(2.304, -1.731),(2.021, -3.073),(2.244, -4.818),(2.396, -0.005),(2.033, -1.505),(2.349, -3.05),(2.393, -4.808),(2.081, 0.014),(2.042, -1.565),(2.369, -3.083),(-0.063, -4.963),(2.23, 0.055),(2.167, -1.388),(4.503, -3.293),(4.323, -4.687),(2.187, -0.118),(4.366, -1.468),(4.353, -3.07),(4.592, -4.775),(1.999, -0.034),(2.251, -1.654),(1.999, -3.222),(2.007, -4.756),(2.198, -0.202),(2.043, -1.485),(2.221, -3.292),(-0.023, -4.93),(2.039, -0.142),(2.268, -1.595),(2.251, -3.256),(-0.157, -4.781),(2.227, -0.027),(2.251, -1.555),(-0.018, -2.993),(-0.156, -4.946),(2.152, 0.125),(2.239, -1.401),(2.144, -3.047),(2.195, -4.994),(2.35, 0.078),(4.258, -1.569),(4.375, -3.325),(4.479, -4.822),(2.088, -0.169),(2.069, -1.703),(2.231, -2.987),(2.33, -4.939),(2.394, -0.003),(2.364, -1.634),(2.385, -3.247),(2.153, -4.981),(3.477, -0.053),(3.212, -1.454),(3.356, -3.105),(1.092, -4.934),(2.28, -0.076),(2.017, -1.502),(2.35, -3.316),(2.135, -4.596),(3.336, -0.154),(3.4, -1.527),(3.311, -3.398),(3.229, -4.636),(3.249, 0.017),(3.238, -1.424),(5.475, -3.017),(5.397, -4.947),(3.252, -0.083),(3.151, -1.519),(3.268, -3.321),(3.103, -4.784),(2.252, -0.088),(2.199, -1.753),(1.997, -3.202),(2.273, -4.946),(2.02, -0.101),(2.269, -1.626),(0.076, -3.102),(2.013, -4.799),(2.123, -0.092),(2.182, -1.78),(2.341, -3.279),(2.143, -4.76),(3.234, -0.113),(1.277, -1.513),(0.978, -3.17),(3.327, -4.689),(3.399, -0.176),(3.403, -1.479),(3.502, -3.189),(1.039, -4.651),(2.218, -0.168),(2.253, -1.702),(2.226, -3.313),(2.309, -4.764)}
            \draw[red,fill=red] \Point circle (0.045ex);
        \foreach \Point in {(5.67, -0.092),(5.527, -1.487),(5.603, -3.154),(5.477, -4.843),(4.465, 0.108),(4.238, -1.722),(2.032, -3.328),(0.035, -4.683),(4.427, -0.175),(2.093, -1.463),(2.113, -3.052),(0.111, -4.865),(5.672, 0.129),(3.203, -1.592),(3.497, -3.138),(1.146, -4.916),(4.474, 0.078),(4.39, -1.708),(4.556, -3.239),(4.335, -4.596),(4.22, -0.098),(4.458, -1.8),(4.453, -3.4),(4.416, -5.019),(5.296, -0.035),(5.455, -1.732),(5.502, -3.101),(3.19, -4.691),(4.398, -0.173),(4.247, -1.46),(4.498, -3.226),(4.367, -4.59),(4.397, 0.213),(4.329, -1.673),(4.264, -3.146),(4.378, -4.621),(4.39, -0.15),(4.503, -1.504),(4.384, -3.3),(4.261, -4.674),(4.25, 0.14),(4.443, -1.491),(4.411, -3.12),(4.5, -4.762),(5.42, -0.173),(5.5, -1.735),(5.397, -3.294),(5.363, -4.796),(4.335, 0.082),(4.409, -1.738),(4.559, -3.183),(4.405, -4.602),(5.491, 0.171),(5.516, -1.387),(5.501, -3.036),(5.483, -4.791),(5.426, -0.2),(5.414, -1.753),(5.537, -3.332),(5.585, -4.846),(4.522, 0.16),(4.53, -1.481),(4.328, -3.092),(4.38, -4.946),(5.708, -0.033),(5.515, -1.381),(5.709, -3.209),(5.582, -4.714),(4.453, 0.06),(4.489, -1.716),(4.47, -3.231),(4.54, -4.833),(4.22, 0.017),(4.41, -1.776),(4.274, -3.1),(4.272, -4.899),(5.586, 0.164),(5.578, -1.406),(5.678, -3.245),(5.522, -4.653),(5.552, 0.142),(5.697, -1.503),(5.411, -3.265),(5.365, -4.967),(5.704, -0.064),(5.62, -1.633),(5.697, -3.275),(5.548, -4.655),(5.663, 0.09),(5.542, -1.52),(5.573, -3.389),(5.498, -4.858),(5.666, -0.035),(5.367, -1.506),(5.55, -3.005),(5.642, -4.895),(4.317, -0.108),(4.235, -1.558),(4.566, -3.182),(4.284, -4.759),(5.585, 0.167),(5.42, -1.576),(5.649, -3.212),(5.637, -4.899),(4.425, 0.092),(4.21, -1.584),(4.31, -3.053),(4.314, -4.942),(5.407, 0.186),(5.508, -1.684),(5.334, -3.092),(5.512, -4.951),(4.471, 0.119),(4.335, -1.786),(4.41, -3.017),(4.198, -4.758),(5.512, 0.083),(5.523, -1.41),(5.559, -3.194),(3.218, -4.708),(4.45, 0.191),(4.437, -1.426),(2.186, -3.377),(0.002, -4.868),(5.649, 0.116),(5.303, -1.559),(5.578, -3.128),(5.452, -4.887),(4.535, -0.019),(4.259, -1.543),(4.498, -3.126),(4.47, -4.901),(5.628, 0.127),(5.573, -1.507),(5.436, -3.147),(5.439, -4.639),(5.554, -0.072),(5.511, -1.475),(5.362, -3.081),(5.288, -4.851),(4.375, 0.046),(4.268, -1.752),(4.327, -3.146),(4.25, -4.957),(4.511, 0.081),(4.358, -1.567),(4.588, -3.225),(4.293, -4.927),(4.52, -0.105),(4.251, -1.648),(4.313, -3.2),(4.456, -4.812),(5.437, -0.096),(5.331, -1.609),(5.517, -3.096),(5.322, -4.701),(4.584, 0.044),(4.539, -1.567),(4.37, -3.28),(4.45, -4.811),(5.578, 0.147),(5.545, -1.576),(5.334, -3.282),(5.441, -4.793),(5.461, -0.11),(5.497, -1.474),(5.361, -3.361),(5.307, -4.725),(4.473, 0.127),(4.567, -1.736),(4.285, -3.322),(4.284, -4.767),(4.351, -0.014),(4.385, -1.505),(4.33, -3.164),(4.483, -4.914),(5.688, -0.112),(5.514, -1.462),(5.513, -3.126),(5.391, -4.655),(5.56, -0.15),(5.462, -1.467),(5.491, -3.068),(5.629, -4.765),(4.448, -0.165),(2.15, -1.511),(0.006, -3.202),(0.044, -4.811),(4.524, 0.077),(4.511, -1.596),(4.441, -3.195),(4.28, -4.821),(5.695, 0.045),(5.339, -1.731),(5.595, -3.018),(5.397, -4.64),(4.321, 0.107),(4.438, -1.607),(4.225, -3.321),(4.501, -4.678),(4.433, -0.07),(4.394, -1.622),(4.358, -3.392),(4.612, -4.792),(4.4, 0.208),(4.262, -1.483),(2.259, -3.099),(-0.086, -4.887),(5.346, -0.155),(5.64, -1.507),(5.536, -3.151),(5.401, -4.86)}
            \draw[cyan,fill=cyan] \Point circle (0.045ex);
    \end{tikzpicture}}
    \caption{Time-evolving random walk on a line graph $ \mc{G}^{(M)} $, where $ \mc{V} = \{\mc[1]{v}, \dots, \mc[6]{v}\} $ and $ M = 4 $. The solid lines represent edges with weight 1, dashed lines represent edges with weight 0.1, and dotted lines represent edges with weights 0.01. At $ t = 1 $, we have 3 clusters of 2 vertices each. Green, red and blue random walkers begin in clusters 1, 2 and 3, respectively. At each time step, the graph updates and the random walkers take one step. In this line graph, we update only the edge between $ \mc[2]{v} $ and $ \mc[3]{v} $, increasing the weight at $ t = 1, 2, 3 $, so that by $ t = 4 $ we have 2 clusters. That is, the clusters $ \{\mc[1]{v}, \mc[2]{v}\} $ and $ \{\mc[3]{v}, \mc[4]{v}\} $ at $ t = 1 $ merge to form a single cluster $ \{\mc[1]{v}, \mc[2]{v}, \mc[3]{v}, \mc[4]{v}\} $ at $t=4$.}
    \label{fig:terw_example}
\end{figure}
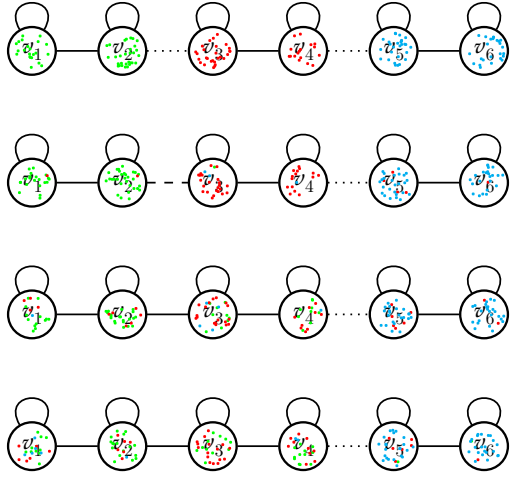

To illustrate random walks on time-evolving graphs, see Figure~\ref{fig:terw_example}. The movement of the random walkers over time shows that sets of vertices that are coherent across the whole time interval (cluster 3) have very few random walkers escaping. Comparatively, the clusters that merge over time initially have little mixing of random walkers between them, but later have a large amount of mixing. This illustrates the notion of coherence and changing cluster structure in time-evolving graphs.
\end{examplewithend}

\subsection{Coherence and Clusters}

When we consider directed graphs, there are two distinct types of cluster that we aim to detect. The first, which we refer to as a \emph{conventional} cluster, or alternatively a \emph{diagonal} cluster, can be described informally as a group of vertices that are more similar to each other than they are to the other vertices in the network. For clustering on graphs we typically consider the weights of edges between vertices to be their similarity or closeness, so a conventional cluster in this case is a group of vertices where there are many edges connecting vertices inside the group, and relatively few edges from these vertices leaving the group. We can also call these \emph{diagonal} clusters, as they can be represented in adjacency matrices as dense blocks on the diagonal, see, for example, Figure \ref{fig:benchmark1_fig}(a).

We are also interested in groups of vertices where the group members all have common in-links or common out-links. That is, they may not be well connected internally but they may instead all share a directed edge to a vertex in another set. We refer to this type of cluster as a \emph{non-conventional} or \emph{off-diagonal} cluster, because these clusters can be represented as off-diagonal blocks in adjacency matrices. These notions are described also in \cite{klus_2022, klus_2024}, where we formulate a definition of \emph{coherence} that is based on \cite{Froyland_2010, Banisch_2017}.

\begin{definition}
    Let $S^{\tau}$ be the flow associated with a dynamical system for some fixed lag time $\tau$. If two sets $\mathbb{A}$ and $\mathbb{B}$ satisfy
    \begin{align*}
        S^{\tau} (\mathbb{A} )\approx \mathbb{B} \text{ and } S^{-\tau}( \mathbb{B}) \approx \mathbb{A},
    \end{align*}
    then we say that $\mathbb{A}$ and $\mathbb{B}$ form a coherent pair and  call $\mathbb{A}$ a \emph{finite-time coherent set}.
\end{definition}

This definition shows that clusters can be regarded as sets that remain almost invariant under the forward--backward dynamics of a system. The spectral clustering approach in this work is based on a graph Laplacian that contains information about these coherent sets and their evolution in time, and we proceed by defining the necessary transfer operators that are required for this.

\subsection{Transfer Operators} \label{sec:transfer_operators}

We now define transfer operators on graphs, in particular the Koopman and Perron--Frobenius operators, which describe the evolution of observables and probability densities of a dynamical system, respectively. These operators can either be estimated from random-walk data or expressed directly using the transition matrices of a graph~\cite{klus_2024}.

\begin{definition}[Probability density] \label{def:probability}
Let $ \mathbb{U} = \{f: \mc{V} \to \mathbb{R}\} $ be the set of real-valued functions defined on $ \mc{V} $. Then a function $ \mu \in \mathbb{U} $ satisfying $ \mu(\mc[i]{v}) \geq 0 $ and
\begin{align*}
    \sum_{i=1}^n \mu(\mc[i]{v}) = 1
\end{align*}
is called a \emph{probability density}.
\end{definition}

\begin{definition}[Perron--Frobenius \& Koopman operators] \label{def:transfer_operators}
Let $ S^{(M)} = \{S_1, \dots, S_M \} $ be the transition probability matrices associated with the time-evolving graph $ \mc{G}^{(M)} = \{ \mc[1]{G}, \dots, \mc[M]{G} \} $.

\begin{enumerate}[leftmargin=3.5ex, itemsep=0ex, topsep=0.5ex, label=\roman*)]
\item For a function $ \rho \in \mathbb{U} $, the \emph{Perron--Frobenius operator} at view $ t $, $ \mathcal{P}_t \colon \mathbb{U} \to \mathbb{U} $, is defined by
\begin{equation*}
    \mathcal{P}_t \ts \rho(\mc[i]{v}) = \sum_{j=1}^n (S_t)_{ji} \ts \rho(\mc[j]{v}).
\end{equation*}
\item Similarly, for a function $ f \in \mathbb{U} $, the \emph{Koopman operator} at time $ t $, $ \mathcal{K}_t \colon \mathbb{U} \to \mathbb{U} $, is given by
\begin{equation*}
    \mathcal{K}_t \ts f(\mc[i]{v}) = \sum_{j=1}^n (S_t)_{ij} \ts f(\mc[j]{v}).
\end{equation*}
\end{enumerate}
\end{definition}

We also define a reweighted Perron--Frobenius operator, which propagates densities with respect to a given reference density.

\begin{definition}[Reweighted Perron--Frobenius operator]\label{def: transfer_operators_reweighted}
Let $ \mu_1 $ be the initial reference density at $ t = 1 $ and define $ \mu_{t+1} = \mathcal{P}_t \ts \mu_t $, i.e.,
\begin{equation*}
    \mu_{t+1}(\mc[i]{v}) = \sum_{j=1}^n (S_t)_{ji} \ts \mu_t(\mc[j]{v}).
\end{equation*}
In what follows, we assume that $ \mu_t(\mc[i]{v}) > 0 $ for all $ t \in \{1, \dots, M\} $. For $ u \in \mathbb{U} $, the \emph{reweighted Perron--Frobenius operator} at time $ t $, $ \mathcal{T}_t \colon \mathbb{U} \to \mathbb{U} $, is defined by
\begin{equation*}
    \mathcal{T}_t \ts u(\mc[i]{v}) = \frac{1}{\mu_{t+1}(\mc[i]{v})} \sum_{j=1}^n (S_t)_{ji} \ts \mu_{t}(\mc[j]{v}) \ts u(\mc[j]{v}).
\end{equation*}
\end{definition}

For each time view $ t \in \{1, \dots, M\} $, the graph $ \mc[t]{G} $ has an associated transition matrix $ S_t $ and also an associated density $ \mu_t $. We choose the initial density $ \mu_1 $ to be the uniform density. As in~\cite{klus_2024}, we define the density matrices
\begin{align*}
    D_{\mu_t} = \diag(\boldsymbol{\mu}_t),
\end{align*}
where $ \boldsymbol{\mu}_t \in \R^n $ is the vector representation of $ \mu_t $. Since we assumed the densities $ \mu_t $ to be strictly positive for all $ t $, the matrices $ D_{\mu_t} $ are invertible. We also note that transfer operators on graphs are of finite dimension, and so we can easily compute their matrix representations. Referring again to \cite{klus_2024}, we define at each time view $ t $
\begin{equation} \label{eq:transition_matrices}
\begin{split}
    \mathcal{K}_t \ts \mathbf{f} &:= K_t \ts \mathbf{f} = S_t \ts \mathbf{f}, \\
    \mathcal{T}_t \ts \mathbf{u} &:= T_t \ts \mathbf{u} = D_{\mu_{t+1}}^{-1} S_t^\top D_{\mu_t} \mathbf{u},
\end{split}
\end{equation}
where $ \mathbf{f}, \mathbf{u} \in \R^n $ are the vector representations of $ f, u \in \mathbb{U} $ and $ K_t, T_t \in \R^{n \times n} $ are the matrix representations of $ \mathcal{K}_t, \mathcal{T}_t $. With a slight abuse of notation the operators are applied component-wise.

\subsection{Covariance and Cross-Covariance Operators}

In order to identify clusters using mCCA, we will also need the closely related covariance and cross-covariance operators.

\begin{definition}[Covariance operators]
Let $ S_t $ be the transition probability matrix and $ \mu_t $ the density at time $ t $. Given functions $ f, g \in \mathbb{U} $, we call $ \mathcal{C}_{tt} \colon \mathbb{U} \to \mathbb{U} $, with
\begin{equation*}
    \mathcal{C}_{tt} \ts f(\mc[i]{v}) = \mu_t(\mc[i]{v}) \ts f(\mc[i]{v})
\end{equation*}
\emph{covariance operator} and $ \mathcal{C}_{t(t+1)} \colon \mathbb{U} \to \mathbb{U} $, with
\begin{equation*}
    \mathcal{C}_{t(t+1)} \ts g(\mc[i]{v}) = \sum_{j=1}^n \mu_t(\mc[i]{v}) \ts (S_t)_{ij} \ts g(\mc[j]{v}),
\end{equation*}
\emph{cross-covariance operator}.
\end{definition}

The covariance of functions and cross-covariance between functions can then be computed using these operators, i.e.,
\begin{align*}
    \var(f) &= \innerprod{f}{\mathcal{C}_{tt} \ts f} = \mathbf{f}^\top C_{tt} \ts \mathbf{f}, \\
    \cov(f, g) &= \innerprod{f}{\mathcal{C}_{t(t+1)} \ts g} = \mathbf{f}^\top C_{t(t+1)} \ts \mathbf{g},
\end{align*}
where $ \mathbf{f}, \mathbf{g} \in \R^n $ are the vector representations of the functions $ f, g \in \mathbb{U} $ and $ C_{tt} = D_{\mu_t} $ and $ C_{t(t+1)} = D_{\mu_t} S_t $ are the matrix representations of the operators $ \mathcal{C}_{tt} $ and $ \mathcal{C}_{t(t+1)} $ (see \cite{klus_2024} for more details). The correlation between the functions $ f $ and $ g $ can thus be computed via
\begin{align*}
    \corr(f, g)
        &= \frac{\cov(f, g)}{\sqrt{\var(f)}{\sqrt{\var(g)}}}
        = \frac{\innerprod{f}{\mathcal{C}_{t(t+1)} \ts g}}{\innerprod{f}{\mathcal{C}_{tt} \ts f}^{\nicefrac{1}{2}}{\innerprod{g}{\mathcal{C}_{(t+1)(t+1)} \ts g}^{\nicefrac{1}{2}}}} \\
        &= \frac{\mathbf{f}^\top C_{t(t+1)} \ts \mathbf{g}}{\big(\mathbf{f}^\top C_{tt} \ts \mathbf{f}\big)^{\nicefrac{1}{2}} \big(\mathbf{g}^\top C_{(t+1)(t+1)} \ts \mathbf{g}\big)^{\nicefrac{1}{2}}}.
\end{align*}

\section{Spectral Clustering of Time-Evolving Graphs} \label{sec:sc_teg}

In this section, we first derive a variant of mCCA in order to define the spatio-temporal graph Laplacian. We then investigate properties of this Laplacian, and propose a spectral clustering algorithm for time-evolving graphs.

\subsection{Multiview CCA} \label{sec:CCA}

Let $ X_t $ be a set of random vectors and define $ x_t = \innerprod{w_t}{X_t} $ to be the projections onto the vectors $ w_t $ for $ t \in \{1, \dots, M\} $. Multiview CCA aims to find the projections given by $ w_t $ that maximize the average correlation
\begin{equation*}
    \sigma = \frac{1}{M-1} \sum_{t=1}^{M-1}{\corr(x_t, x_{t+1})}.
\end{equation*}
We choose this coupling of time views since we are interested in detecting the evolution of clusters throughout a time interval, but other time-couplings can also be chosen. In \cite{cristianini_2004}, for example, a multiview CCA formulation that maximizes the correlation across \emph{all} views is proposed.

Typically, in order to detect coherent sets, we estimate the required correlations from Lagrangian data (e.g., GPS tracking data if we want to identify gyres in the ocean \cite{FSvS14}). In the graph setting, however, we can compute the correlations using the covariance and cross-covariance operators as shown above. Let $ f_t \in \mathbb{U} $ now be a function at view $ t $. Our goal is to maximize the average correlation
\begin{align*}
    \sigma &= \frac{1}{M-1} \sum_{t=1}^{M-1} \corr(f_t, f_{t+1}) \\
           &= \frac{1}{M-1} \sum_{t=1}^{M-1} \frac{\mathbf{f}_t^\top C_{t(t+1)} \ts \mathbf{f}_{t+1}}{\big(\mathbf{f}_t^\top C_{tt} \ts \mathbf{f}_t\big)^{\nicefrac{1}{2}} \big(\mathbf{f}_{t+1}^\top C_{(t+1)(t+1)} \ts \mathbf{f}_{t+1}\big)^{\nicefrac{1}{2}}}.
\end{align*}
Note that $ \mathbf{f}_t $ is only determined up to scalar multiplication since
\begin{align*}
    \frac{\lambda_t \ts \mathbf{f}_t^\top C_{t(t+1)} \ts \lambda_{t+1} \ts \mathbf{f}_{t+1}}{\big(\lambda_t^2 \ts \mathbf{f}_t^\top C_{tt} \ts \mathbf{f}_t\big)^{\nicefrac{1}{2}} \big(\lambda_{t+1}^2 \mathbf{f}_{t+1}^\top C_{(t+1)(t+1)} \ts \mathbf{f}_{t+1}\big)^{\nicefrac{1}{2}}}
    = \frac{\mathbf{f}_t^\top C_{t(t+1)} \ts \mathbf{f}_{t+1}}{\big(\mathbf{f}_t^\top C_{tt} \ts \mathbf{f}_t\big)^{\nicefrac{1}{2}} \big(\mathbf{f}_{t+1}^\top C_{(t+1)(t+1)} \ts \mathbf{f}_{t+1}\big)^{\nicefrac{1}{2}}},
\end{align*}
where $ \lambda_t, \lambda_{t+1} > 0 $ are constants. We can therefore constrain the variance terms to equal one and, omitting the constant prefactor $ \frac{1}{M-1} $, write the mCCA problem as
\begin{align*}
    &\max_{\mathbf{f}_t} \sum_{t=1}^{M-1} \mathbf{f}_t^\top C_{t(t+1)} \ts \mathbf{f}_{t+1}, \\
    &\text{s.t. } \mathbf{f}_t^\top C_{tt} \ts \mathbf{f}_t = 1 ~ \forall \ts t = 1, \dots, M.
\end{align*}
We use Lagrange multipliers to obtain
\begin{align*}
    \mathcal{L}(\mathbf{f}_1, \lambda_1, \dots, \mathbf{f}_M, \lambda_M) &= \sum_{t=1}^{M-1} \mathbf{f}_t^\top C_{t(t+1)} \ts \mathbf{f}_{t+1} - \tfrac{1}{2} \lambda_1 (\mathbf{f}_1^\top C_{11} \ts \mathbf{f}_1 - 1) \\& - \sum_{t=2}^{M-1} \lambda_t (\mathbf{f}_t^\top C_{tt} \ts \mathbf{f}_t - 1) - \tfrac{1}{2} \lambda_M (\mathbf{f}_M^\top C_{MM} \ts \mathbf{f}_M - 1),
\end{align*}
where the factors $ \frac{1}{2} $ for the Lagrange multipliers $ \lambda_1 $ and $ \lambda_M $ are chosen for convenience. We then compute the gradients with respect to each view, which gives the following set of equations:
\begin{align*}
    C_{12} \ts \mathbf{f}_2 - \lambda_1 \ts C_{11} \ts \mathbf{f}_1 &= 0, \\
    C_{t(t-1)} \ts \mathbf{f}_{t-1} + C_{t(t+1)} \ts \mathbf{f}_{t+1} - 2 \ts \lambda_t \ts C_{tt} \ts \mathbf{f}_t &= 0, \quad \text{for } t = 2, \dots, M-1, \\
    C_{M(M-1)} \ts \mathbf{f}_{M-1} - \lambda_M \ts C_{MM} \ts \mathbf{f}_M &= 0.
\end{align*}
Assuming that $ \lambda_t = \lambda $ for all $ t $, we can write the set of equations given above as the generalized eigenvalue problem
\begin{align*}
     C_{12} \ts \mathbf{f}_2 &= \lambda \ts C_{11} \ts \mathbf{f}_1, \\
    C_{t(t-1)} \ts \mathbf{f}_{t-1} + C_{t(t+1)} \ts \mathbf{f}_{t+1} &= 2 \ts \lambda \ts C_{tt} \ts \mathbf{f}_t, \quad \text{for } t= 2, \dots, M-1, \\
    C_{M(M-1)} \ts \mathbf{f}_{M-1} &= \lambda \ts C_{MM} \ts \mathbf{f}_M,
\end{align*}
which we write in matrix form as
\begin{equation} \label{eq:mcca_gen_eigenproblem}
    \mathbf{A} \ts \mathbf{f} = \lambda \ts \mathbf{B} \ts \mathbf{f},
\end{equation}
with
\begin{align*}
    \mathbf{A} =
    \begin{bmatrix}
        0 & C_{12} & 0 & \dots & 0 \\
        C_{21} & 0 & \ddots & \ddots & \vdots \\
        0 & \ddots & 0 & \ddots & 0 \\
        \vdots & \ddots & \ddots & \ddots & \hspace{-1em}C_{(M-1)M} \\
        0 & \dots & 0 & C_{M(M-1)} & 0 \\
    \end{bmatrix}\!, ~
    \mathbf{B} =
    \begin{bmatrix}
        C_{11} & 0 & \dots & \dots & 0 \\
        0 & 2 \ts C_{22} & \ddots & \ddots & \vdots \\
        \vdots & \ddots & \ddots & \ddots & \vdots \\
        \vdots & \ddots & \ddots & 2 \ts C_{(M-1)(M-1)} \hspace{-1em} & 0 \\
        0 & \dots & \dots & 0 & C_{MM}
    \end{bmatrix}.
\end{align*}
Since $ C_{tt} = D_{\mu_t} $ and we assume all densities $ \mu_t $ to be positive, the matrix $ \mathbf{B} $ is invertible. Furthermore, $ C_{t(t+1)} = D_{\mu_t} S_t $ implies that $ C_{t(t-1)} = C_{(t-1)t}^\top = S_{t-1}^\top D_{\mu_{t-1}} $. Defining the matrix $ \mathbf{C} := \mathbf{B}^{-1}\mathbf{A} $, we obtain the eigenvalue problem
\begin{align} \label{eq:mcca_eigenproblem}
    \mathbf{C} \ts \mathbf{f} = \lambda \ts \mathbf{f},
\end{align}
with
\begin{equation*}
    \mathbf{C} =
    \begin{bmatrix}
        0 & S_1 & 0 & \dots & 0 \\
        \frac{1}{2}D_{\mu_2}^{-1}S_1^\top D_{\mu_1} & 0 & \frac{1}{2} S_2 & \ddots & \vdots \\
        0 & \ddots & 0 & \ddots & 0 \\
        \vdots & \ddots & \ddots & \ddots & \frac{1}{2}S_{M-1}\\
        0 & \dots & 0 & D_{\mu_M}^{-1} S_{M-1}^\top D_{\mu_{M-1}} & 0
    \end{bmatrix}.
\end{equation*}
The matrix $ \mathbf{C} $ can be reformulated in terms of transfer operators. Using the matrix representations described in \eqref{eq:transition_matrices}, we have
\begin{equation} \label{eq:CM_transfer}
    \mathbf{C} =
    \begin{bmatrix}
        0 & K_1 & 0 & \dots & 0 \\
        \frac{1}{2}T_{1} & 0 & \frac{1}{2}K_2 & \ddots & \vdots \\
        0 & \ddots & 0 & \ddots & 0 \\
        \vdots & \ddots & \frac{1}{2}T_{M-2} & \ddots & \frac{1}{2}K_{M-1} \\
        0 & \dots & 0 & T_{M-1} & 0 \\
    \end{bmatrix}.
\end{equation}

\begin{examplewithend}
In the special case $ M = 2 $, mCCA is exactly standard CCA. This is straightforward as, choosing $ M = 2 $, we obtain the following maximization problem:
\begin{align*}
    &\max_{\mathbf{f}_t}  \mathbf{f}_1^\top C_{12} \ts \mathbf{f}_2 \\
    &\text{s.t. } \mathbf{f}_t^\top C_{tt} \ts \mathbf{f}_t = 1 ~\text{for } t = 1,2.
\end{align*}
Following the derivation above, we can write this as the generalized eigenvalue problem
\begin{align*}
    \begin{bmatrix}
        0 & C_{12} \\
        C_{21} & 0 \\
    \end{bmatrix}
    \begin{bmatrix}
        \mathbf{f}_1 \\ \mathbf{f}_2
    \end{bmatrix}
    = \lambda
    \begin{bmatrix}
        C_{11} & 0 \\
        0 & C_{22} \\
    \end{bmatrix}
    \begin{bmatrix}
        \mathbf{f}_1 \\ \mathbf{f}_2
    \end{bmatrix}.
\end{align*}
From these equations we can show that $ \sigma = \lambda $ by multiplying the first equation by $ \mathbf{f}_1^\top $ since
\begin{equation*}
    \underbrace{\mathbf{f}_1^\top C_{12} \ts \mathbf{f}_{2}}_{=\sigma} = \lambda \underbrace{\mathbf{f}_1^\top C_{11} \ts \mathbf{f}_1}_{=1} = \lambda.
\end{equation*}
Hence, the correlation $ \sigma $ is maximized by the eigenvector corresponding to the largest eigenvalue. We can rewrite the above eigenvalue problem as
\begin{equation*}
    C_{11}^{-1} C_{12} \ts C_{22}^{-1} C_{21} \ts \mathbf{f}_1 = \lambda^2 \ts \mathbf{f}_1 \implies K_1 \ts T_1 \ts \mathbf{f}_1 = \lambda^2 \ts \mathbf{f}_1,
\end{equation*}
which illustrates that $ \mathbf{f}_1 $ is almost invariant under the forward--backward dynamics if there is an eigenvalue $ \lambda \approx 1 $.
\end{examplewithend}

\begin{remark}
    The matrix $\mathbf{C}$ is related to the matrix $T_{rev}$ described in \cite{fackeldey_2019}, although the derivation is notably different. In \cite{fackeldey_2019}, the authors propose a spectral analysis of non-autonomous transfer operators in order to detect coherence in time-dependent dynamical systems. These transfer operators are conceptually similar to those defined in Definitions \ref{def:transfer_operators} and \ref{def: transfer_operators_reweighted}, and in \cite{fackeldey_2019} a discretization in space and time of the operator $\mathcal{T}^{\tau}$ leads to a matrix $T_{rev}$ of dimension $(nM, nM)$, which is exactly the dimension of $\mathbf{C}$. The forward-backward construction of $T_{rev}$ from a time-dependent matrix $T(t) \in \mathbb{R}^{n \times n}$ can be interpreted similarly to $\mathbf{C}$ as written in \eqref{eq:CM_transfer}. Both matrices describe push-forward and pull-back dynamics in the system; in our work, we show that this construction can be formulated via an optimization problem, and in \cite{fackeldey_2019} it is instead motivated by considering a number of time-discretizations and their physical interpretation in the system.
    \end{remark}

\subsection{The Spatio-Temporal Graph Laplacian}

We now use mCCA to define a spatio-temporal graph Laplacian and analyze the eigenspectrum. Additionally, we highlight connections to transfer operators defined in Section~\ref{sec:background}.

\begin{definition}[Spatio-temporal graph Laplacian]
We call the matrix $ \mathbf{L} = I - \mathbf{C} $ the \emph{spatio-temporal graph Laplacian}.
\end{definition}

\begin{proposition} \label{prop:properties}
The matrix $ \mathbf{C} $ has the following properties:
\begin{enumerate}[leftmargin=4.5ex, itemsep=0ex, topsep=0.5ex, label=\roman*)]
\item The eigenvalues of $ \mathbf{C} $ are contained in the interval $ [-1, 1] $.
\item The spectrum of $ \mathbf{C} $ is symmetric about zero.
\end{enumerate}
\end{proposition}

\begin{proof} ~
\begin{enumerate}[wide, itemindent=2ex, itemsep=0ex, topsep=0.5ex, label=\roman*)]
\item We first show that the eigenvalues are real-valued. It is well-known that the eigenvalues of the generalized eigenproblem~\eqref{eq:mcca_gen_eigenproblem} are real-valued if $ \mathbf{A} $ and $ \mathbf{B} $ are both symmetric and $ \mathbf{B} $ is also positive definite. Since $ C_{ij} = C_{ji}^\top $, the matrix $ \mathbf{A} $ is symmetric. Also, $ \mathbf{B} $ is a diagonal matrix whose diagonal entries are positive since $ \mu_t $ is by assumption strictly positive for all~$ t $. Therefore, $ \mathbf{B} $ is positive definite and the eigenvalues of \eqref{eq:mcca_gen_eigenproblem} and thus also \eqref{eq:mcca_eigenproblem} are real-valued. Next, we prove that the spectral radius of $ \mathbf{C} $ is 1, which we do by showing that the matrix is row-stochastic. This is straightforward using the fact that $ K_t \mathbbm{1} = T_t \mathbbm{1} = \mathbbm{1} $ for all $1 \leq t \leq M$, where $\mathbbm{1} $ is a vector of ones.
\item Let $ \lambda $ be an eigenvalue of $ \mathbf{C} $ with associated eigenvector $ \mathbf{f} = [\mathbf{f}_1^\top, \dots, \mathbf{f}_M^\top]^\top $, that is, $ \mathbf{C} \ts \mathbf{f} = \lambda \ts \mathbf{f} $. Define $ \mathbf{g} = [\mathbf{f}_1^\top, -\mathbf{f}_2^\top, \mathbf{f}_3^\top, -\mathbf{f}_4^\top, \dots]^\top $. Using \eqref{eq:CM_transfer}, it follows that
\begin{equation*}
    \mathbf{C} \ts \mathbf{g} =
    \begin{bmatrix}
        -K_1 \ts \mathbf{f}_2 \\
        \frac{1}{2} T_1 \ts \mathbf{f}_1 + \frac{1}{2} K_2 \ts \mathbf{f}_3 \\[0.5ex]
        -\frac{1}{2} T_2 \ts \mathbf{f}_2 - \frac{1}{2} K_3 \ts \mathbf{f}_4 \\[0.5ex]
        \frac{1}{2} T_3 \ts \mathbf{f}_3 + \frac{1}{2} K_4 \ts \mathbf{f}_5 \\
        \vdots
    \end{bmatrix}
    =
    \begin{bmatrix}
        -\lambda \ts \mathbf{f}_1 \\[0.5ex]
        \lambda \ts \mathbf{f}_2 \\[0.5ex]
        -\lambda \ts \mathbf{f}_3 \\[0.5ex]
        \lambda \ts \mathbf{f}_4 \\
        \vdots
    \end{bmatrix}
    = - \lambda
    \begin{bmatrix}
        \mathbf{f}_1 \\[0.5ex]
        -\mathbf{f}_2 \\[0.5ex]
        \mathbf{f}_3 \\[0.5ex]
        -\mathbf{f}_4 \\ \vdots
    \end{bmatrix} = - \lambda \mathbf{g}.
\end{equation*}
Therefore, $ \mathbf{g} $ is also an eigenvector of $ \mathbf{C} $ with associated eigenvalue $ -\lambda $. \qedhere
\end{enumerate}
\end{proof}

Note that we are only interested in the positive eigenvalues since negative eigenvalues correspond to negatively correlated functions.

\begin{remark} \label{remark:mcca_transfer}
Assume that $ \lambda \approx 1 $. For the special case $ M = 2 $, this means that we find functions that are almost invariant under the forward--backward dynamics, see \cite{klus_2024}. For the general case, the interpretation of the matrix $ \mathbf{C} $ is similar. The first equation yields $ K_1 \mathbf{f}_2 \approx \mathbf{f}_1 $, i.e., pulling $ \mathbf{f}_2 $ back using the Koopman operator must be highly correlated with $ \mathbf{f}_1 $. The following equations, $ \frac{1}{2} T_{t-1} \mathbf{f}_{t-1} + \frac{1}{2} K_t \mathbf{f}_{t+1} \approx \mathbf{f}_t $ implies that if we push $ \mathbf{f}_{t-1} $ forward and pull $ \mathbf{f}_{t+1} $ back and then average, we should approximately obtain $ \mathbf{f}_{t} $. The last equation suggests that pushing $ \mathbf{f}_{M-1} $ forward should result in $ \mathbf{f}_{M} $.
\end{remark}

\begin{corollary} \label{cor:L_CM}
The eigenvalues of the spatio-temporal graph Laplacian $ \mathbf{L} $ are contained in the interval $ [0, 2] $ and the spectrum is symmetric about $ 1 $.
\end{corollary}

\begin{proof}
This follows immediately from Proposition~\ref{prop:properties} since
\begin{equation*}
    \mathbf{C} \ts \mathbf{f} = \lambda \ts \mathbf{f} \implies \mathbf{L} \ts \mathbf{f} = (I - \mathbf{C}) \ts \mathbf{f} = (1 - \lambda) \ts \mathbf{f}. \qedhere
\end{equation*}
\end{proof}

Using Corollary~\ref{cor:L_CM}, it is clear that we can maximize the correlation between the functions~$ f_t $ by computing the eigenvector corresponding to the smallest eigenvalue of the spatio-temporal graph Laplacian. Let us now investigate the connections between the spatio-temporal graph Laplacian $ \mathbf{L} $ and the original time-evolving graph by analyzing the structure of the matrix $ \mathbf{C} = \mathbf{B}^{-1} \mathbf{A} $. Given a time-evolving graph $ \mc{G}^{(M)} = \{\mc[1]{G}, \dots, \mc[M]{G}\} $ with vertices $ \mc{V} = \{\mc[1]{v}, \dots, \mc[n]{v}\} $, we define a static graph $ \Wtilde = (\Vtilde, \Etilde) $ using the matrix $ \mathbf{A} $ associated with $ \mc{G}^{(M)} $ as an adjacency matrix. We refer to $ \Wtilde $ as a \emph{coupling graph}, as it couples all time views into a single representation. This coupling graph contains $ M $ copies of the $ n $ vertices, so $ \Vtilde = \big\{\mc[1]{v}^{(1)}, \dots, \mc[n]{v}^{(1)}, \dots, \mc[1]{v}^{(M)}, \dots, \mc[n]{v}^{(M)}\big\} $.

\begin{lemma} \label{lemma:mcca_structure}
Given $ \mc{G}^{(M)} $ with associated coupling graph $\Wtilde$, there exists an edge
\begin{equation*}
    \big(\mc[\alpha]{v}^{(t)}, \mc[\beta]{v}^{(t+1)}\big) \in \, \Etilde
\end{equation*}
if and only if there exists an edge $ (\mc[\alpha]{v}, \mc[\beta]{v}) $ in $ \mc[t]{G} $. Further, the edge in $\Wtilde$ is always undirected.
\end{lemma}

\begin{proof}
Recalling the definition of $\mathbf{A}$ from \eqref{eq:mcca_gen_eigenproblem} and assuming that $\mu_t$ is strictly positive for all $t$, we have
\begin{align*}
    (\mc[\alpha]{v}^{(t)}, \mc[\beta]{v}^{(t+1)}) \in \Wtilde &\iff \mathbf{A}_{(n \cdot t + \alpha)(n \cdot (t+1) +\beta)} = (C_{t(t+1)})_{\alpha \beta} = (D_{\mu_t} S_t)_{\alpha\beta} \neq 0 \\
    &\iff (S_t)_{\alpha\beta} \neq 0 \\
    &\iff (\mc[\alpha]{v}, \mc[\beta]{v}) \in \mc[t]{G}.
\end{align*}
Now, note that since $C_{t(t+1)}^\top = C_{(t+1)t}$, the matrix $\mathbf{A}$ is symmetric. Therefore an edge in $\Wtilde$ is necessarily undirected as we construct the graph using $\mathbf{A}$ as the adjacency matrix.
\end{proof}

Constructing a coupling graph in this way, it is clear that cluster structure in the original time-evolving graph is preserved. The new graph contains edges between copies of vertices at adjacent time views, so we can detect the changing nature of the clusters over time as these connections appear or disappear from one view to the next.

\begin{examplewithend} \label{example2}
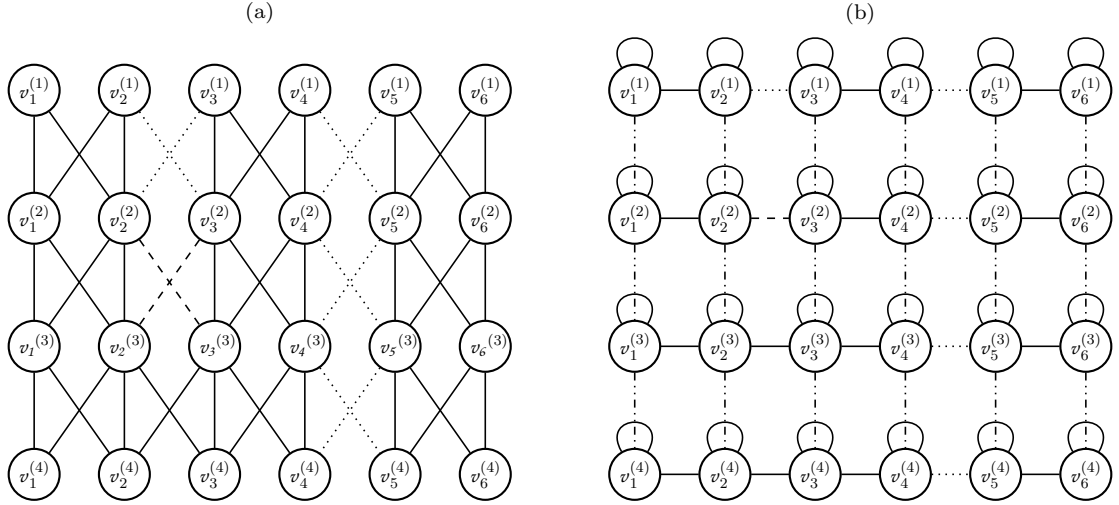
\begin{figure}
    \centering
    \begin{minipage}[b]{0.48\textwidth}
        \centering
        \subfiguretitle{(a)}
        \vspace{3ex}
        \begin{tikzpicture}[
                >= stealth, 
                semithick 
            ]
            \tikzstyle{every state}=[
                draw=black,
                thick,
                fill=white,
                inner sep=0pt,
                text width=6mm,
                align=center,
                scale=0.7
            ]
            \tikzset{every loop/.style={}} 
            \node[state] (v1) {$\mc[1]{v}^{(1)}$};
            \node[state] (v2) [right=0.5cm of v1] {$\mc[2]{v}^{(1)}$};
            \node[state] (v3) [right=0.5cm of v2] {$\mc[3]{v}^{(1)}$};
            \node[state] (v4) [right=0.5cm of v3] {$\mc[4]{v}^{(1)}$};
            \node[state] (v5) [right=0.5cm of v4] {$\mc[5]{v}^{(1)}$};
            \node[state] (v6) [right=0.5cm of v5] {$\mc[6]{v}^{(1)}$};
            \node[state] (v7) [below=1cm of v1]{$\mc[1]{v}^{(2)}$};
            \node[state] (v8) [right=0.5cm of v7] {$\mc[2]{v}^{(2)}$};
            \node[state] (v9) [right=0.5cm of v8] {$\mc[3]{v}^{(2)}$};
            \node[state] (v10) [right=0.5cm of v9] {$\mc[4]{v}^{(2)}$};
            \node[state] (v11) [right=0.5cm of v10] {$\mc[5]{v}^{(2)}$};
            \node[state] (v12) [right=0.5cm of v11] {$\mc[6]{v}^{(2)}$};

            \path[-] (v1) edge node {} (v7);
            \path[-] (v2) edge node {} (v8);
            \path[-] (v3) edge node {} (v9);
            \path[-] (v4) edge node {} (v10);
            \path[-] (v5) edge node {} (v11);
            \path[-] (v6) edge node {} (v12);

            \path[-] (v1) edge node {} (v8);
            \path[-] (v2) edge node {} (v7);
            \path[-, dotted] (v2) edge node {} (v9);
            \path[-, dotted] (v3) edge node {} (v8);
            \path[-] (v3) edge node {} (v10);
            \path[-] (v4) edge node {} (v9);
            \path[-, dotted] (v4) edge node {} (v11);
            \path[-, dotted] (v5) edge node {} (v10);
            \path[-] (v5) edge node {} (v12);
            \path[-] (v6) edge node {} (v11);

            \node[state] (v13) [below=1cm of v7]{$\mc{v_1}^{(3)}$};
            \node[state] (v14) [right=0.5cm of v13] {$\mc{v_2}^{(3)}$};
            \node[state] (v15) [right=0.5cm of v14] {$\mc{v_3}^{(3)}$};
            \node[state] (v16) [right=0.5cm of v15] {$\mc{v_4}^{(3)}$};
            \node[state] (v17) [right=0.5cm of v16] {$\mc{v_5}^{(3)}$};
            \node[state] (v18) [right=0.5cm of v17] {$\mc{v_6}^{(3)}$};

            \path[-] (v7) edge node {} (v13);
            \path[-] (v8) edge node {} (v14);
            \path[-] (v9) edge node {} (v15);
            \path[-] (v10) edge node {} (v16);
            \path[-] (v11) edge node {} (v17);
            \path[-] (v12) edge node {} (v18);

            \path[-] (v7) edge node {} (v14);
            \path[-] (v8) edge node {} (v13);
            \path[-, dashed] (v8) edge node {} (v15);
            \path[-, dashed] (v9) edge node {} (v14);
            \path[-] (v9) edge node {} (v16);
            \path[-] (v10) edge node {} (v15);
            \path[-, dotted] (v10) edge node {} (v17);
            \path[-, dotted] (v11) edge node {} (v16);
            \path[-] (v11) edge node {} (v18);
            \path[-] (v12) edge node {} (v17);

            \node[state] (v19) [below=1cm of v13] {$\mc[1]{v}^{(4)}$};
            \node[state] (v20) [right=0.5cm of v19] {$\mc[2]{v}^{(4)}$};
            \node[state] (v21) [right=0.5cm of v20] {$\mc[3]{v}^{(4)}$};
            \node[state] (v22) [right=0.5cm of v21] {$\mc[4]{v}^{(4)}$};
            \node[state] (v23) [right=0.5cm of v22] {$\mc[5]{v}^{(4)}$};
            \node[state] (v24) [right=0.5cm of v23] {$\mc[6]{v}^{(4)}$};

            \path[-] (v13) edge node {} (v19);
            \path[-] (v14) edge node {} (v20);
            \path[-] (v15) edge node {} (v21);
            \path[-] (v16) edge node {} (v22);
            \path[-] (v17) edge node {} (v23);
            \path[-] (v18) edge node {} (v24);

            \path[-] (v13) edge node {} (v20);
            \path[-] (v14) edge node {} (v19);
            \path[-] (v14) edge node {} (v21);
            \path[-] (v15) edge node {} (v20);
            \path[-] (v15) edge node {} (v22);
            \path[-] (v16) edge node {} (v21);
            \path[-, dotted] (v16) edge node {} (v23);
            \path[-, dotted] (v17) edge node {} (v22);
            \path[-] (v17) edge node {} (v24);
            \path[-] (v18) edge node {} (v23);
        \end{tikzpicture}
    \end{minipage}
    \hfill
    \begin{minipage}[b]{0.48\textwidth}
    \centering
        \subfiguretitle{(b)}
        \vspace{-0.7ex}
        \begin{tikzpicture}[
                >= stealth, 
                semithick 
            ]
            \tikzstyle{every state}=[
                draw=black,
                thick,
                fill=white,
                inner sep=0pt,
                text width=6mm,
                align=center,
                scale=0.7
            ]
            \tikzset{every loop/.style={}} 

            \node[state] (v25) {$\mc[1]{v}^{(1)}$};
            \node[state] (v26) [right=0.5cm of v25] {$\mc[2]{v}^{(1)}$};
            \node[state] (v27) [right=0.5cm of v26] {$\mc[3]{v}^{(1)}$};
            \node[state] (v28) [right=0.5cm of v27] {$\mc[4]{v}^{(1)}$};
            \node[state] (v29) [right=0.5cm of v28] {$\mc[5]{v}^{(1)}$};
            \node[state] (v30) [right=0.5cm of v29] {$\mc[6]{v}^{(1)}$};

            \path[-] (v25) edge node {} (v26);
            \path[-, dotted] (v26) edge node {} (v27);
            \path[-] (v27) edge node {} (v28);
            \path[-, dotted] (v28) edge node {} (v29);
            \path[-] (v29) edge node {} (v30);

            \path (v25) edge [loop above, min distance=6mm, in=60, out=120] node {} (v25);
            \path (v26) edge [loop above, min distance=6mm, in=60, out=120] node {} (v26);
            \path (v27) edge [loop above, min distance=6mm, in=60, out=120] node {} (v27);
            \path (v28) edge [loop above, min distance=6mm, in=60, out=120] node {} (v28);
            \path (v29) edge [loop above, min distance=6mm, in=60, out=120] node {} (v29);
            \path (v30) edge [loop above, min distance=6mm, in=60, out=120] node {} (v30);

            \node[state] (v31) [below=1cm of v25]{$\mc[1]{v}^{(2)}$};
            \node[state] (v32) [right=0.5cm of v31] {$\mc[2]{v}^{(2)}$};
            \node[state] (v33) [right=0.5cm of v32] {$\mc[3]{v}^{(2)}$};
            \node[state] (v34) [right=0.5cm of v33] {$\mc[4]{v}^{(2)}$};
            \node[state] (v35) [right=0.5cm of v34] {$\mc[5]{v}^{(2)}$};
            \node[state] (v36) [right=0.5cm of v35] {$\mc[6]{v}^{(2)}$};

            \path[-] (v31) edge node {} (v32);
            \path[-, dashed] (v32) edge node {} (v33);
            \path[-] (v33) edge node {} (v34);
            \path[-, dotted] (v34) edge node {} (v35);
            \path[-] (v35) edge node {} (v36);

            \path (v31) edge [loop above, min distance=6mm, in=60, out=120] node {} (v31);
            \path (v32) edge [loop above, min distance=6mm, in=60, out=120] node {} (v32);
            \path (v33) edge [loop above, min distance=6mm, in=60, out=120] node {} (v33);
            \path (v34) edge [loop above, min distance=6mm, in=60, out=120] node {} (v34);
            \path (v35) edge [loop above, min distance=6mm, in=60, out=120] node {} (v35);
            \path (v36) edge [loop above, min distance=6mm, in=60, out=120] node {} (v36);

            \path[-, dash dot dot] (v25) edge node {} (v31);
            \path[-, dash dot dot] (v26) edge node {} (v32);
            \path[-, dash dot dot] (v27) edge node {} (v33);
            \path[-, dash dot dot] (v28) edge node {} (v34);
            \path[-, dash dot dot] (v29) edge node {} (v35);
            \path[-, dash dot dot] (v30) edge node {} (v36);

            \node[state] (v37) [below=1cm of v31]{$\mc[1]{v}^{(3)}$};
            \node[state] (v38) [right=0.5cm of v37] {$\mc[2]{v}^{(3)}$};
            \node[state] (v39) [right=0.5cm of v38] {$\mc[3]{v}^{(3)}$};
            \node[state] (v40) [right=0.5cm of v39] {$\mc[4]{v}^{(3)}$};
            \node[state] (v41) [right=0.5cm of v40] {$\mc[5]{v}^{(3)}$};
            \node[state] (v42) [right=0.5cm of v41] {$\mc[6]{v}^{(3)}$};

            \path[-] (v37) edge node {} (v38);
            \path[-] (v38) edge node {} (v39);
            \path[-] (v39) edge node {} (v40);
            \path[-, dotted] (v40) edge node {} (v41);
            \path[-] (v41) edge node {} (v42);

            \path (v37) edge [loop above, min distance=6mm, in=60, out=120] node {} (v37);
            \path (v38) edge [loop above, min distance=6mm, in=60, out=120] node {} (v38);
            \path (v39) edge [loop above, min distance=6mm, in=60, out=120] node {} (v39);
            \path (v40) edge [loop above, min distance=6mm, in=60, out=120] node {} (v40);
            \path (v41) edge [loop above, min distance=6mm, in=60, out=120] node {} (v41);
            \path (v42) edge [loop above, min distance=6mm, in=60, out=120] node {} (v42);

            \path[-, dash dot dot] (v31) edge node {} (v37);
            \path[-, dash dot dot] (v32) edge node {} (v38);
            \path[-, dash dot dot] (v33) edge node {} (v39);
            \path[-, dash dot dot] (v34) edge node {} (v40);
            \path[-, dash dot dot] (v35) edge node {} (v41);
            \path[-, dash dot dot] (v36) edge node {} (v42);

            \node[state] (v43) [below=1cm of v37] {$\mc[1]{v}^{(4)}$};
            \node[state] (v44) [right=0.5cm of v43] {$\mc[2]{v}^{(4)}$};
            \node[state] (v45) [right=0.5cm of v44] {$\mc[3]{v}^{(4)}$};
            \node[state] (v46) [right=0.5cm of v45] {$\mc[4]{v}^{(4)}$};
            \node[state] (v47) [right=0.5cm of v46] {$\mc[5]{v}^{(4)}$};
            \node[state] (v48) [right=0.5cm of v47] {$\mc[6]{v}^{(4)}$};

            \path[-] (v43) edge node {} (v44);
            \path[-] (v44) edge node {} (v45);
            \path[-] (v45) edge node {} (v46);
            \path[-, dotted] (v46) edge node {} (v47);
            \path[-] (v47) edge node {} (v48);

            \path (v43) edge [loop above, min distance=6mm, in=60, out=120] node {} (v43);
            \path (v44) edge [loop above, min distance=6mm, in=60, out=120] node {} (v44);
            \path (v45) edge [loop above, min distance=6mm, in=60, out=120] node {} (v45);
            \path (v46) edge [loop above, min distance=6mm, in=60, out=120] node {} (v46);
            \path (v47) edge [loop above, min distance=6mm, in=60, out=120] node {} (v47);
            \path (v48) edge [loop above, min distance=6mm, in=60, out=120] node {} (v48);

            \path[-, dash dot dot] (v37) edge node {} (v43);
            \path[-, dash dot dot] (v38) edge node {} (v44);
            \path[-, dash dot dot] (v39) edge node {} (v45);
            \path[-, dash dot dot] (v40) edge node {} (v46);
            \path[-, dash dot dot] (v41) edge node {} (v47);
            \path[-, dash dot dot] (v42) edge node {} (v48);
        \end{tikzpicture}
    \end{minipage}
    \caption{(a) Coupling graph with adjacency matrix $ \mathbf{A} $. Vertices are connected to their copies in adjacent time views because of the presence of self-loops in the original graph, and the changing cluster structure of the original graph is visible. (b) Coupling graph associated with the supra-Laplacian. Here, the edge weights between vertices and their copies in adjacent time views have weight $ a $, represented by the dashed-dotted line.}
    \label{fig:mcca_edges2}
\end{figure}

We illustrate the structure of the coupling graph $ \Wtilde $ associated with the spatio-temporal graph Laplacian of the line graph introduced in Example~\ref{example1}, where self-loops are added to the line graph for regularization of the coupling graph. Figure~\ref{fig:mcca_edges2}\ts(a) shows the edges in the graph, and the time evolution of the clusters can be observed.
\end{examplewithend}

\subsection{Comparison with the Supra-Laplacian}
As a comparison, we briefly describe the supra-Laplacian defined as in~\cite{gomez_2013}. This approach constructs a Laplacian by computing a graph Laplacian at each snapshot, and then using these matrices as diagonal blocks for a matrix of size $ (n \cdot M) \times (n \cdot M) $. This matrix is referred to as the \emph{supra-Laplacian}, and the sub- and super-diagonal $ n \times n $ blocks of this matrix contain copies of $ a I_n $, where $ a $ is the coupling strength between each snapshot. For simplicity we use a constant coupling parameter for every snapshot, but note that in general the coupling can be time-dependent.
\begin{examplewithend}
We continue with Example~\ref{example2} and Figure~\ref{fig:mcca_edges2}\ts(b) illustrates the relationships between layers of the time-evolving line graph when the layers are coupled using coupling parameter $ a $. A clear difference between the supra-Laplacian coupling graph and the spatio-temporal graph Laplacian coupling graph is that an edge $ (\mc[\alpha]{v}, \mc[\beta]{v}) $ in $ \mc[t]{G} $ translates to an edge $ (\mc[n \cdot t + \alpha]{v}, \mc[n \cdot t + \beta]{v}) $ in the static coupling graph associated with the supra-Laplacian.
\end{examplewithend}

It is well-known that, using this approach, the eigenvalues and eigenvectors are strongly dependent on the choice of coupling parameter~\cite{gomez_2013}. In general, a strong coupling between time layers leads to a spatial clustering that detects intra-layer vertex clusters, and in the limit of large values of $ a $, this leads to a temporal aggregation of the network. On the other hand, a weak coupling leads to a temporal clustering, where vertices are clustered together because they exist in the same time view, rather than belonging to the same community over time. Therefore, the parameter $ a $ must be tuned and the optimal value is highly problem-specific. The spatio-temporal graph Laplacian, on the other hand, does not require any parameter tuning. We use $ \widehat{\mc{G}} $ to denote the coupling graph associated with the supra-Laplacian, and write $ \mathbf{L}_{S} $ for the supra-Laplacian matrix.

\begin{remark} \label{remark:mCCA_directed}
Constructing a coupling graph $\Wtilde$ using the matrix $ \mathbf{A} $ associated with the spatio-temporal graph Laplacian always produces an undirected graph, even if $ \mc{G}^{(M)} $ is directed (see Lemma~\ref{lemma:mcca_structure}). In contrast, if $ \mc{G}^{(M)} $ is directed, then the coupling graph $ \widehat{\mc{G}} $ associated with the supra-Laplacian is directed. The asymmetry of $ \mathbf{L}_S $ means that the eigenvalues are in general complex-valued, and standard spectral clustering algorithms will fail~\cite{klus_2024_review}. In order to apply spectral clustering, some form of symmetrization or other procedure is required. This shows that the spatio-temporal graph Laplacian can naturally cluster directed time-evolving graphs without losing directionality information which is a key advantage over the supra-Laplacian for spectral clustering.
\end{remark}

\subsection{Spectral Clustering for Time-Evolving Graphs} \label{sec:algorithm}

The spectral clustering approach to community discovery on graphs is based on clustering the eigenvectors associated with the smallest eigenvalues of an associated graph Laplacian. There are many different matrices that are defined in the literature as graph Laplacians, including the unnormalized graph Laplacian and the random-walk Laplacian~\cite{luxburg_2007}. Extensions to spectral clustering capable of community detection in directed graphs include various symmetrization techniques~\cite{satuluri_2011}, complex-valued representations of the graph~\cite{cucuringu_2020}, and a transfer operator-based Laplacian describing the forward--backward dynamics on a graph~\cite{klus_2024}.

In a similar way to these methods, we use the spatio-temporal graph Laplacian to detect clusters in time-evolving graphs. As shown in Corollary~\ref{cor:L_CM}, computing the smallest eigenvalues of this Laplacian $ \mathbf{L} $ is equivalent to computing the largest eigenvalues of $ \mathbf{C} $. The eigenvalues of $ \mathbf{L} $ and $ \mathbf{C} $ are always real-valued, even when the graph itself is directed, and hence we are able to apply spectral clustering also to directed time-evolving graphs.

\begin{algorithm}[Spatio-temporal graph Laplacian spectral clustering algorithm] \label{alg:mCCA_sc} ~
\begin{enumerate}[topsep=1ex, itemsep=0ex]
\item Compute the $ k $ largest eigenvalues $ \lambda_\ell $ and associated eigenvectors $ \boldsymbol{\varphi}_\ell $ of $ \mathbf{C}$.
\item Define $ \boldsymbol{\Phi} = [\boldsymbol{\varphi}_1, \dots, \boldsymbol{\varphi}_k] \in \R^{Mn \times k} $ and let $ r_i $ denote the $ i $th row of $ \boldsymbol{\Phi} $.
\item Cluster the points $ \{ \ts r_i \ts \}_{i=1}^{Mn} $ using, e.g., $ k $-means.
\end{enumerate}
\end{algorithm}

Note that for static graphs, the number of eigenvalues of the transition matrix close to 1 indicates the number of clusters, $ k $, to detect, with a spectral gap between $ \lambda_k $ and $ \lambda_{k+1} $. However, the eigenvalues of $ \mathbf{C} $ contain both spatial information and temporal information, meaning that it is not clear how to choose $ k $ in this case. This is illustrated in Section~\ref{sec:num_results}. For this reason, Algorithm~\ref{alg:mCCA_sc} requires $k$ to be chosen based on existing knowledge of the specific problem, and at this stage cannot be selected based on eigenvalues alone. Applying $ k $-means to $ \mathbf{C} $ then clusters each vertex in each time view, which we can visualize across the whole interval. This allows us to evaluate the performance of the algorithm on graphs without a ground truth clustering.

\section{Numerical Results} \label{sec:num_results}

We generate weighted time-evolving graphs by first constructing static graphs with heterogeneous node-degree distributions and community sizes as described in~\cite{fortunato_2009} using publicly available \href{https://github.com/andrealancichinetti/LFRbenchmarks}{C\texttt{++} code} (Package 4). We then modify and evolve the graphs as described below so that they exhibit the behavior we are interested in.

To analyze the capabilities of the spatio-temporal graph Laplacian, we show the eigenvectors of the matrix $ \mathbf{C} $ along with the dominant eigenvalues. We also visualize the cluster labels generated by applying $ k $-means for vertices across the entire time interval, and evaluate the clustering by computing the \emph{adjusted Rand index} (ARI) \cite{hubert_1985} on the first time view and the final time view only.

\subsection{Benchmark 1}

\begin{figure}
    \definecolor{matlab1}{RGB}{253, 231, 37}
    \definecolor{matlab2}{RGB}{33, 145, 140}
    \definecolor{matlab3}{RGB}{68, 1, 84}
    \newcommand{\cdash}[1]{\textcolor{#1}{\rule[0.5ex]{1em}{0.3ex}}}
    \centering
    \begin{minipage}[b]{\textwidth}
        \centering
        \subfiguretitle{(a)}
        \vspace{1ex}
        \includegraphics[width=\textwidth]{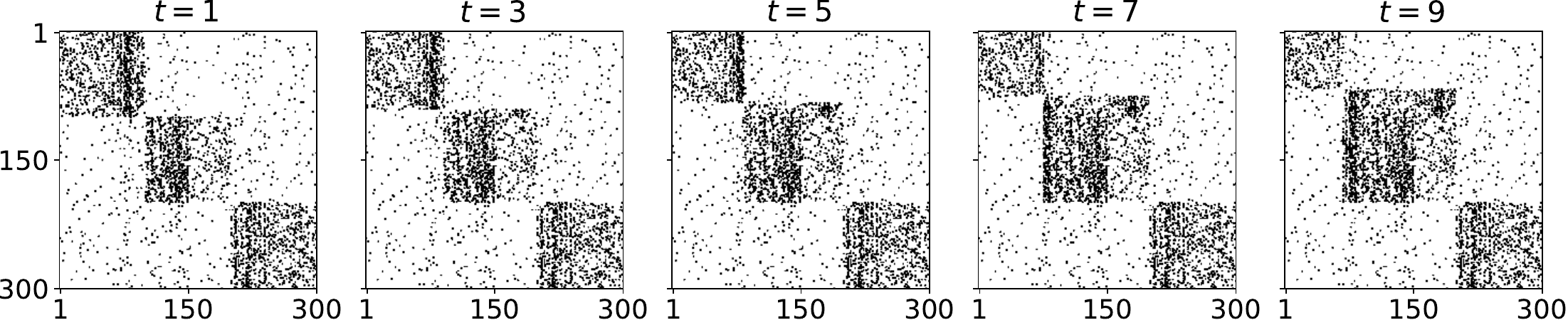}
    \end{minipage}
    \\[1.5ex]
    \begin{minipage}[b]{\textwidth}
        \centering
        \subfiguretitle{(b)}
        \includegraphics[width=\textwidth]{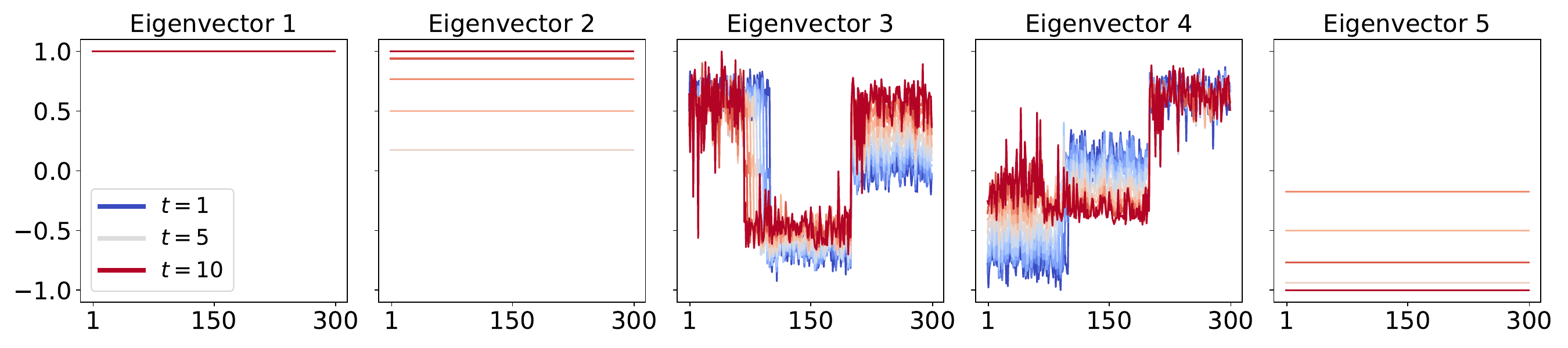}
    \end{minipage}
        \begin{minipage}[t]{0.317\textwidth}
        \centering
        \subfiguretitle{(c)}
        \vspace{0.3ex}
        \includegraphics[width=\textwidth]{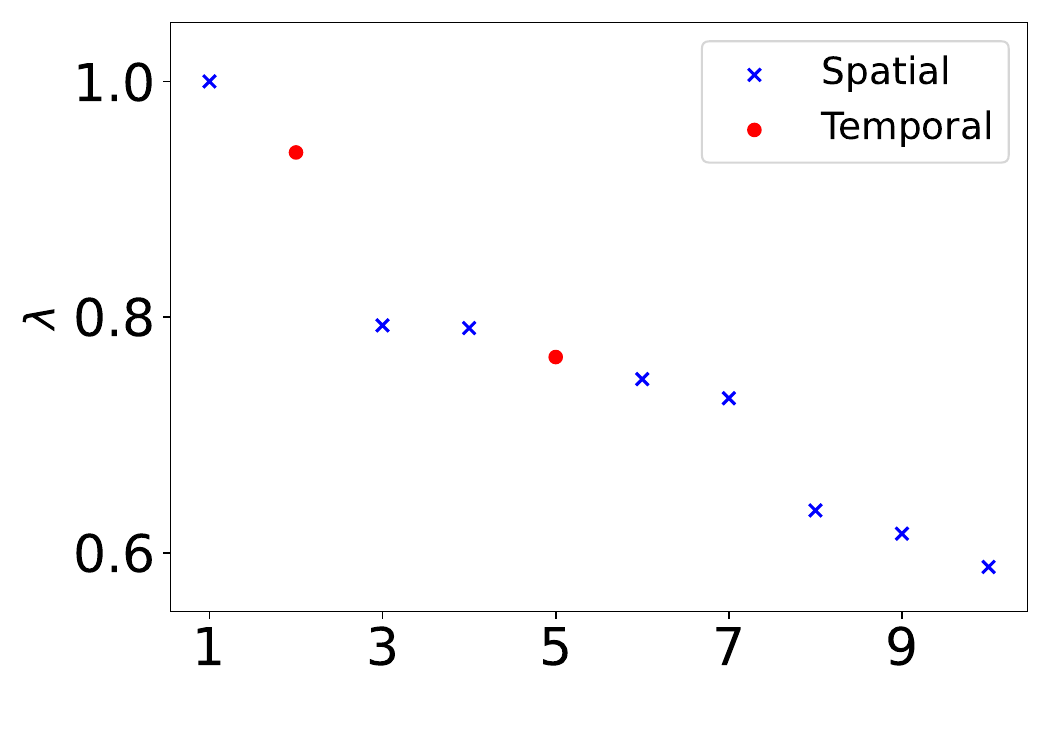}
    \end{minipage}
    \begin{minipage}[t]{0.31\textwidth}
        \centering
        \subfiguretitle{(d)}
        \includegraphics[width=\textwidth]{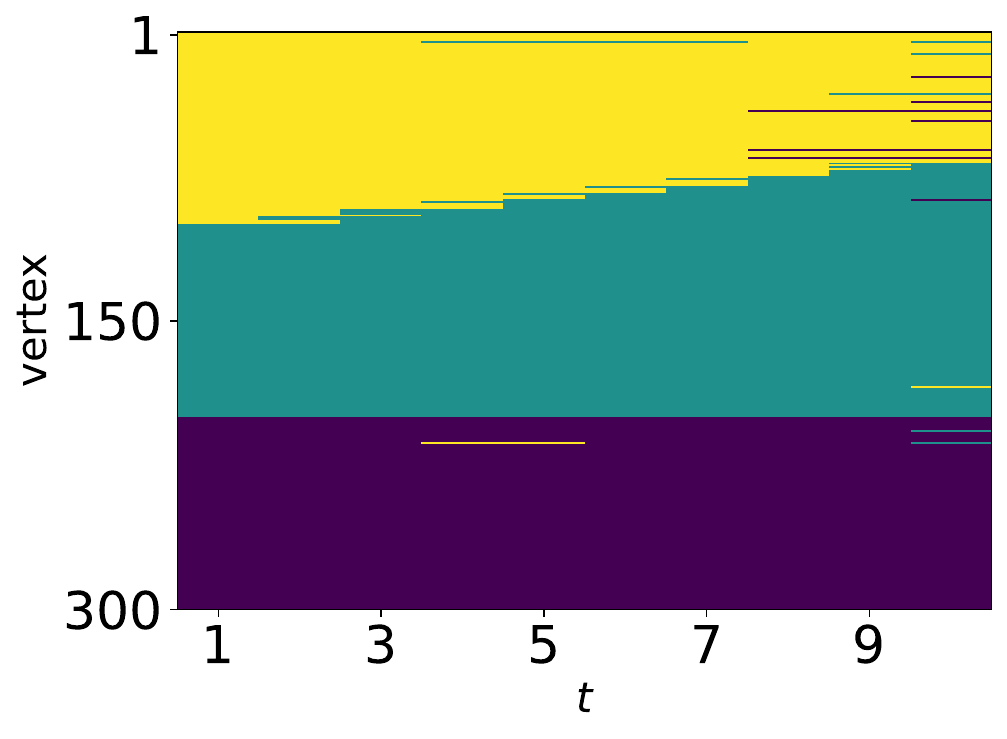}
    \end{minipage}
     \begin{minipage}[t]{0.31\textwidth}
        \centering
        \subfiguretitle{(e)}
        \includegraphics[width=\textwidth]{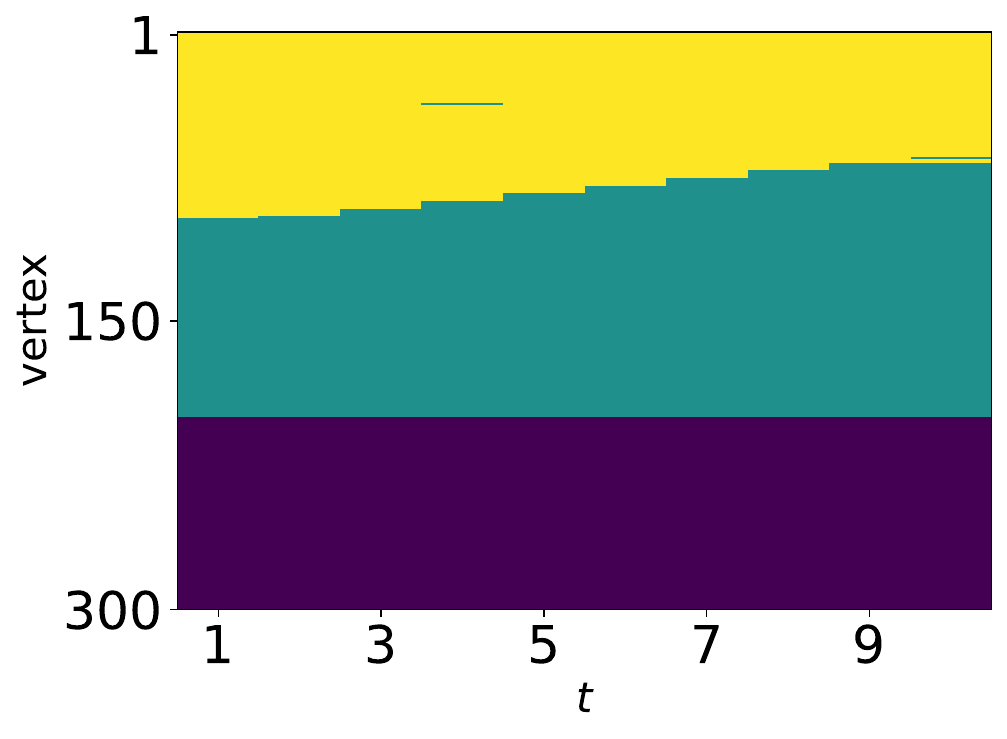}
    \end{minipage}
    \caption{(a) Adjacency matrix snapshots of time-evolving undirected benchmark graph at $ t = 1, 3, 5, 7, 9 $. The graph comprises 3 clusters of $ 100 $ vertices at $ t = 1 $, and over time the graph evolves so that the first cluster shrinks to $ 65 $ vertices and the second cluster grows to $135$ vertices. (b)~The first five dominant eigenvectors of the matrix $\mathbf{C}$ over time. (c)~Ten dominant eigenvalues of $ \mathbf{C} $, showing which eigenvalues encode temporal information and which encode spatial information. We show also the $ k $-means clustering of the graph with $ k = 3 $ clusters using the first three spatial eigenvectors of (d)~the spatio-temporal graph Laplacian and (e)~the supra-Laplacian. In both (d) and (e), \cdash{matlab1} denotes cluster~1, \cdash{matlab2} denotes cluster~2, and \cdash{matlab3} denotes cluster~3.}
    \label{fig:benchmark1_fig}
\end{figure}

Benchmark 1 is constructed by first generating a graph of $ n = 300 $ vertices, in three clusters of 100 vertices each. Then, we move a fraction of the vertices in cluster 1 to cluster 2 at each time step, so that at $t=10$ cluster 1 contains 65 vertices and cluster 2 has 135. Cluster 3 remains constant for the entire time period. This is illustrated in Figure~\ref{fig:benchmark1_fig}\ts(a). Figure~\ref{fig:benchmark1_fig}\ts(b) shows the dominant eigenvectors of the matrix $\mathbf{C}$ for this benchmark. Here, we ``fold'' the eigenvectors so that we display the $ M \times n $ eigenvector as $ M $ vectors of shape $ 1 \times n $. The first of these $ 1 \times n $ vectors gives information about the cluster structure at time $ t = 1 $, the second gives information about the structure at $ t = 2 $, and so on. As we go from $ t = 1 $ to $ t = M $, we plot the vectors from blue to red. It is clear from the eigenvectors that we observe the behavior exhibited by the graph; eigenvector 3 in particular shows the shrinking of cluster 1 as the level set for vertices 1--100 at $ t = 1 $ shrinks to a level set for only the first 65 of these vertices by $ t = 10 $. Similarly, eigenvector 3 also shows that cluster 2 grows. We note here that the first eigenvector is constant as expected, but eigenvectors 2 and 5 are constant within each time view, and change value from one time view to the next. This indicates an eigenvector that is encoding temporal information, rather than spatial. That is, applying a clustering algorithm using this eigenvector would group the vertices into their snapshots (i.e., a temporal clustering), not into their intra-layer spatial clusters. Figure~\ref{fig:benchmark1_fig}\ts(c) shows that these spatial and temporal eigenvalues do not yield a spectral gap as expected for standard spectral clustering methods.

Applying Algorithm~\ref{alg:mCCA_sc} with $ k = 3 $ to benchmark 1, we visualize cluster labels generated by $ k $-means in Figure~\ref{fig:benchmark1_fig}\ts(d) where we have clustered the graph using eigenvectors 1, 3 and 4 of $\mathbf{C}$ (i.e.,\ we filter out the temporal eigenvectors such as 2 and 5, but not the constant first eigenvector by convention). For comparison purposes, we also show the results of clustering the same graph using the supra-Laplacian with $ a = 0.3 $ in Figure~\ref{fig:benchmark1_fig}\ts(e), where we have also filtered out temporal eigenvectors. We see that the matrix associated with the spatio-temporal graph Laplacian is capable of achieving a good clustering over all time views, with some mislabeled vertices occurring at $ t =10 $. In contrast, we must tune the coupling parameter for the supra-Laplacian in order to achieve a similar result. Using a trial--and--error approach we are able to achieve a good clustering with $ a = 0.3 $; Table~\ref{tab:Benchmark graphs} shows that, with this parameter value, the supra-Laplacian outperforms Algorithm~\ref{alg:mCCA_sc} in terms of the ARI at $ t =10 $. However, we note that there is a significant computational cost to tuning the coupling parameter, and it is also not obvious against which metric this parameter should be tuned. For example, we may tune the parameter according to the ARI score at time views for which a ground truth is available, or by visual inspection of the cluster labels. In this work we tune the parameter by considering the ARI at $ t = 1 $ and $ t = 10 $ and by inspection of the intermediate labels.

\subsection{Benchmark 2}

\begin{figure}
    \definecolor{matlab4}{RGB}{253, 231, 37}
    \definecolor{matlab5}{RGB}{53, 183, 121}
    \definecolor{matlab6}{RGB}{49, 104, 142}
    \definecolor{matlab7}{RGB}{68, 1, 84}
    \newcommand{\cdash}[1]{\textcolor{#1}{\rule[0.5ex]{1em}{0.3ex}}}
    \centering
    \begin{minipage}[b]{\textwidth}
        \centering
        \subfiguretitle{(a)}
        \vspace{1ex}
        \includegraphics[width=\textwidth]{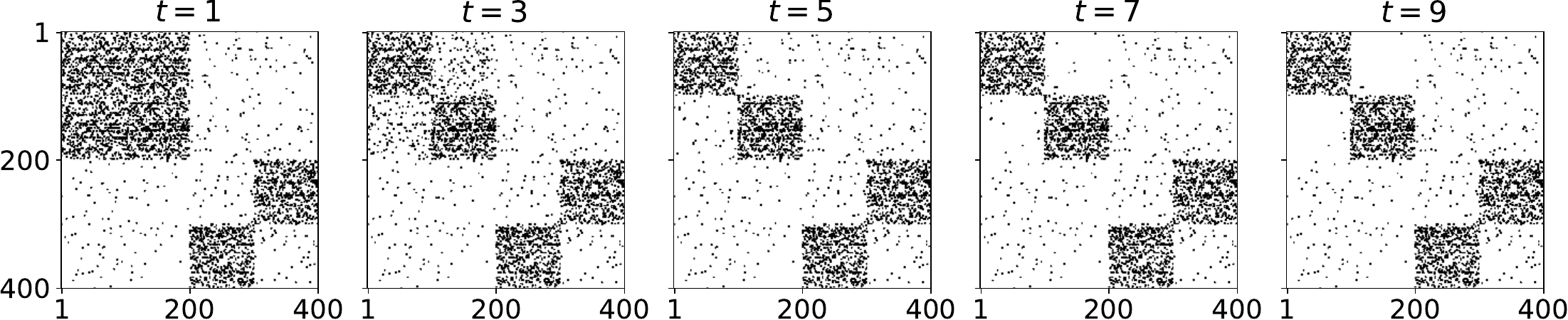}
    \end{minipage}
    \\[1.5ex]
    \begin{minipage}[b]{\textwidth}
        \centering
        \subfiguretitle{(b)}
        \vspace{0.2ex}
        \includegraphics[width=\textwidth]{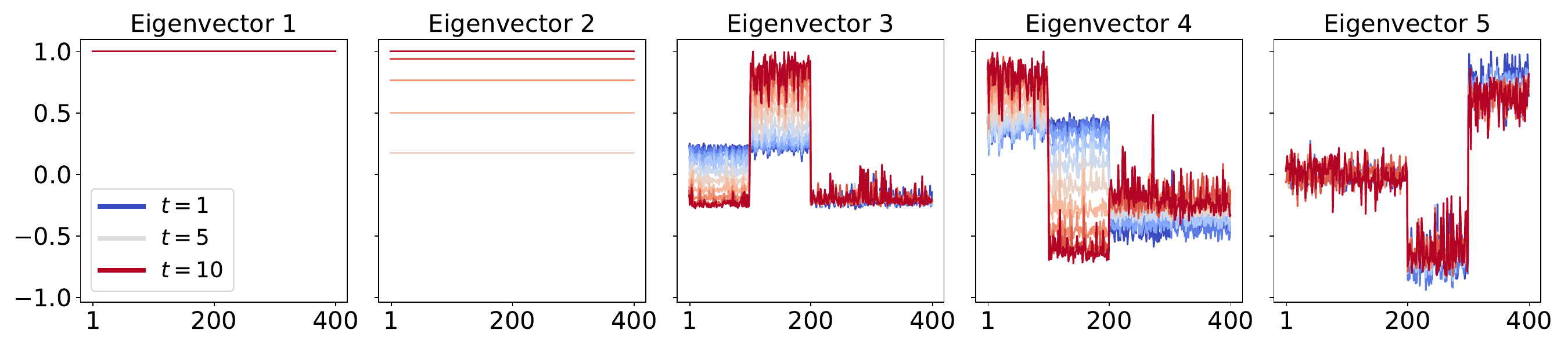}
    \end{minipage}
    \begin{minipage}[t]{0.322\textwidth}
        \centering
        \subfiguretitle{(c)}
        \vspace{0.2ex}
        \includegraphics[width=\textwidth]{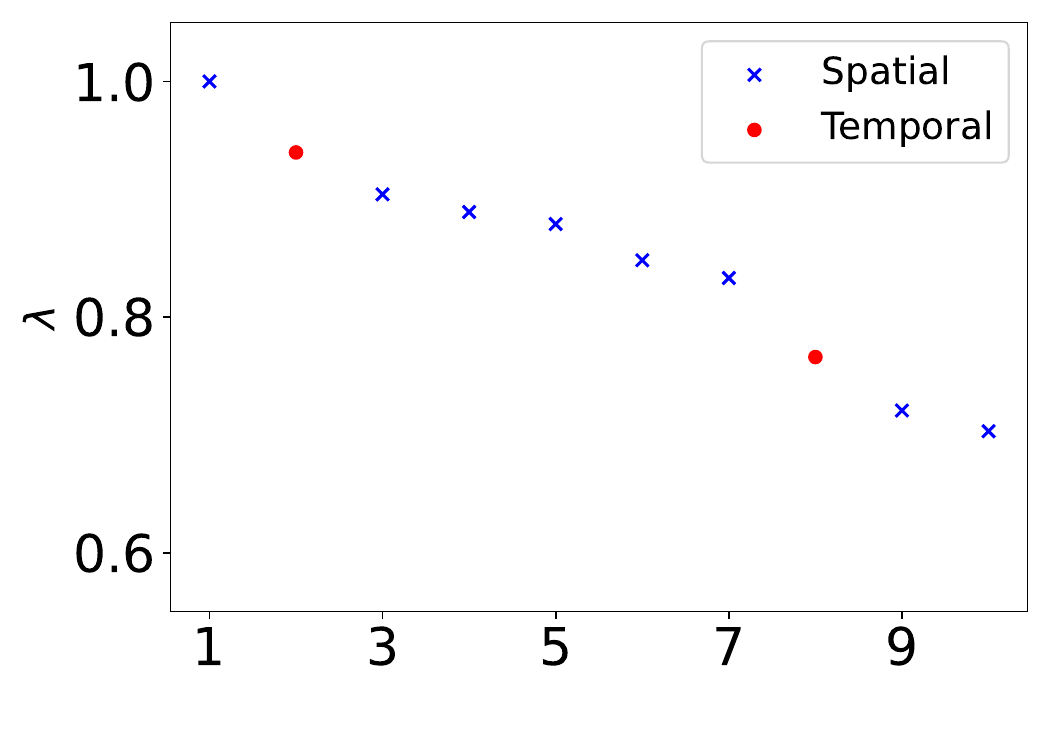}
    \end{minipage}
    \begin{minipage}[t]{0.31\textwidth}
        \centering
        \subfiguretitle{(d)}
        \includegraphics[width=\textwidth]{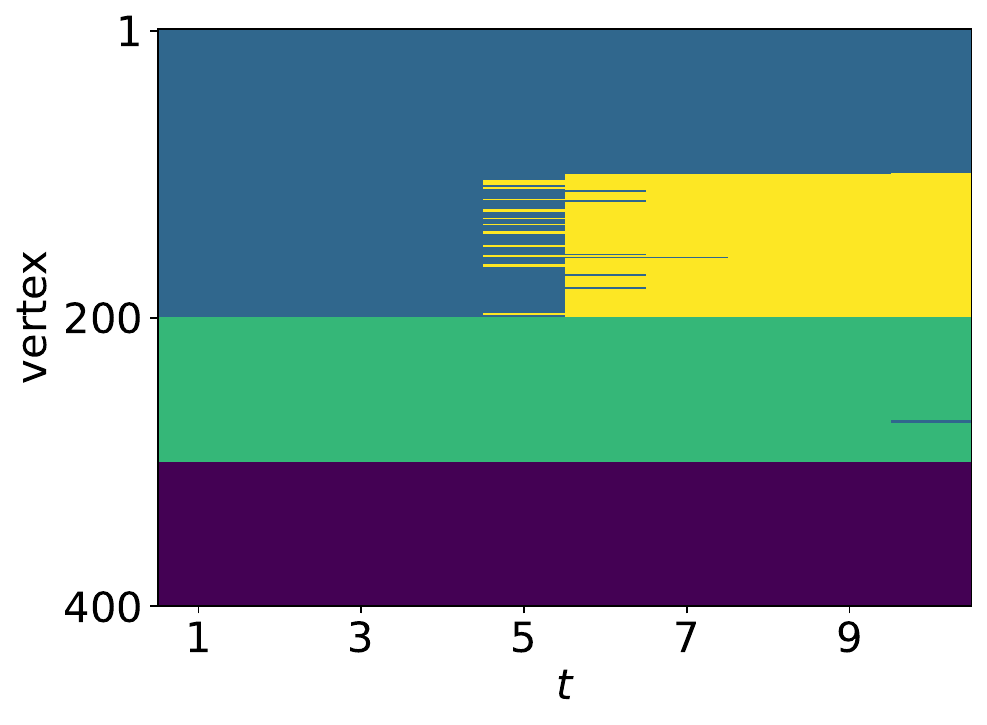}
    \end{minipage}
     \begin{minipage}[t]{0.31\textwidth}
        \centering
        \subfiguretitle{(e)}
        \includegraphics[width=\textwidth]{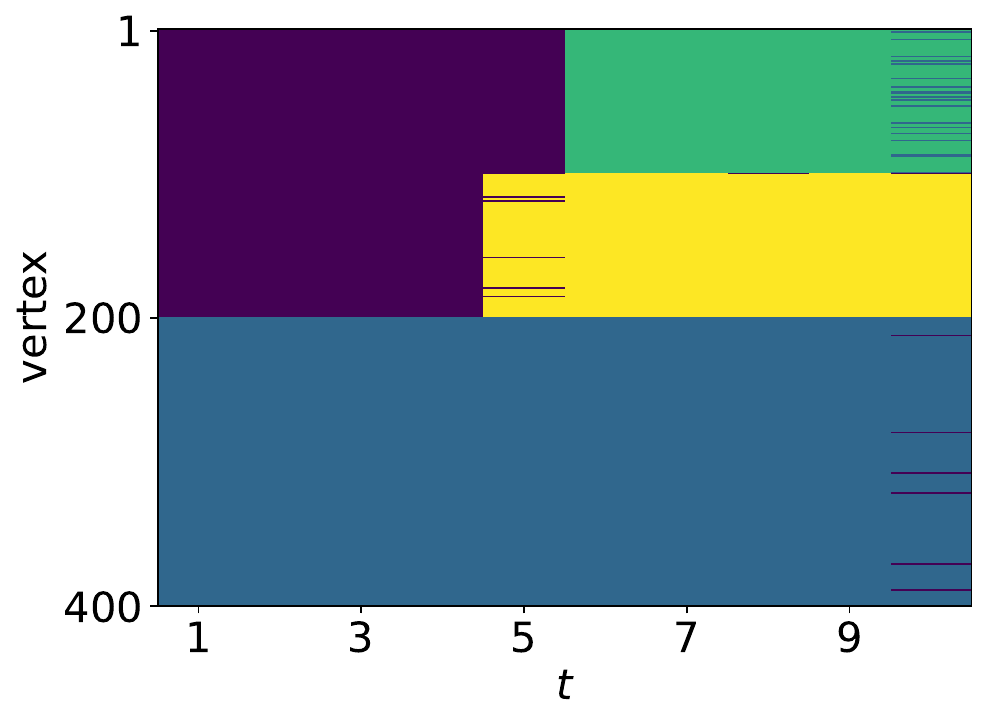}
    \end{minipage}

    \caption{(a) Adjacency matrix snapshots of directed benchmark graph at $ t = 1, 3, 5, 7, 9 $. The graph comprises 3 clusters at $ t = 1 $, where the first cluster has $ 200 $ vertices and the other 2 clusters each have $ 100 $. Over time the graph evolves so that the first cluster splits into two clusters of equal size, so that at $ t = 10 $ each cluster has $ 100 $ vertices. (b) Five dominant eigenvectors of the matrix $ \mathbf{C} $ over time. (c)~Ten dominant eigenvalues of $ \mathbf{C} $, showing which eigenvalues encode temporal information and which encode spatial information. We also show the $ k $-means clustering of the graph with $ k = 3 $ clusters using the first three spatial eigenvectors of (d) the spatio-temporal graph Laplacian and (e) the supra-Laplacian. In both, vertices colored by \cdash{matlab4} are labeled as cluster 1, whilst \cdash{matlab5} denotes cluster~2, \cdash{matlab6} denotes cluster~3, and \cdash{matlab7} denotes cluster~4.}
    \label{fig:benchmark2_fig}
\end{figure}

For the second benchmark, we generate a directed graph comprising 400 vertices and initially three clusters, with a mix of both off- and on-diagonal dense blocks (see~\cite{klus_2024}). We split the largest cluster by removing edges with probability $ p =0.5 $ at each time step over $ t =10 $ views. The adjacency matrix of this graphs at a number of time views is shown in Figure~\ref{fig:benchmark2_fig}\ts(a). We conduct a similar analysis for benchmark 2. Again, it is clear from Figure~\ref{fig:benchmark2_fig}\ts(b) that the eigenvectors of $ \mathbf{C} $ detect the changing cluster structure of the graph, where the largest cluster of $ 200 $ vertices splits into two smaller clusters of equal size. We highlight the spatial and temporal eigenvalues in Figure~\ref{fig:benchmark2_fig}\ts(c), and then apply Algorithm~\ref{alg:mCCA_sc} with $ k = 4 $. Figure~\ref{fig:benchmark2_fig} also shows the cluster labels for vertices over all time views for both Algorithm~\ref{alg:mCCA_sc}, in (d), and also for the supra-Laplacian with $ a = 0.05 $, in (e). Again, we choose $ a $ based on a trial--and--error approach to obtain the cluster labels closest to the true behavior of the graph, and for both approaches we filter out temporal eigenvectors and choose the first $ k $ spatial eigenvectors to cluster.

We note that the eigenvalues of the supra-Laplacian for directed graphs are not necessarily real-valued, and so spectral clustering will fail. In order to use the supra-Laplacian on benchmark 2, we remove directionality information and symmetrize the graph so that it is undirected. We therefore do not expect the supra-Laplacian to be able to detect off-diagonal clusters, like the final two clusters in Figure~\ref{fig:benchmark2_fig}\ts(a). Indeed, this is observed in Figure~\ref{fig:benchmark2_fig}\ts(e), where the $ k $-means labeling fails to identify those two clusters and instead groups them into a single cluster, labeled in Figure~\ref{fig:benchmark2_fig}\ts(e) as cluster 3 for the duration of the time interval. In comparison, the labels generated by Algorithm~\ref{alg:mCCA_sc} correctly identify these off-diagonal blocks and label them as clusters 2 and 4, and also identifies the split in the initial large cluster (labeled as cluster 3 at $ t = 1 $) into two smaller clusters at approximately $ t = 4 $ where half of the vertices retain the label of cluster~3, and half are labeled as cluster~1. We note that the supra-Laplacian approach also detects this splitting cluster, but identifies two new clusters, cluster 1 and 2, emerging at $ t = 4 $ and $ t = 5 $, rather than preserving the original label of cluster 4 of the largest cluster before the split.

\begin{remark}
    In this benchmark graph, we chose $k=4$ since we know that there are $4$ clusters in the graph. However, two of the clusters are off-diagonal and so would be considered as a single cluster by the supra-Laplacian approach as it cannot detect such clusters. For comparative purposes, we also compared the spatio-temporal graph Laplacian with the supra-Laplacian for $k=3$ and $k=5$. The full results are omitted for brevity but in the case $k=3$, the spatio-temporal approach fails to detect the splitting of the first cluster but correctly identifies the off-diagonal blocks, whilst the supra-Laplacian cannot detect the off-diagonal clusters and also mislabels a significant number of vertices after the split. For $k=5$, the spatio-temporal approach is very similar to $k=4$ but generates new labels for both clusters that result from the largest cluster splitting. The supra-Laplacian again suffers from significant mislabeling, and clusters swap labels multiple times. We note that we again tuned the coupling parameter using trial-and-error to achieve the best possible labeling.
\end{remark}

\subsection{Double Gyre Application}
\label{sec:double_gyre}

\begin{figure}
    \centering
    \includegraphics[width=0.8\textwidth]{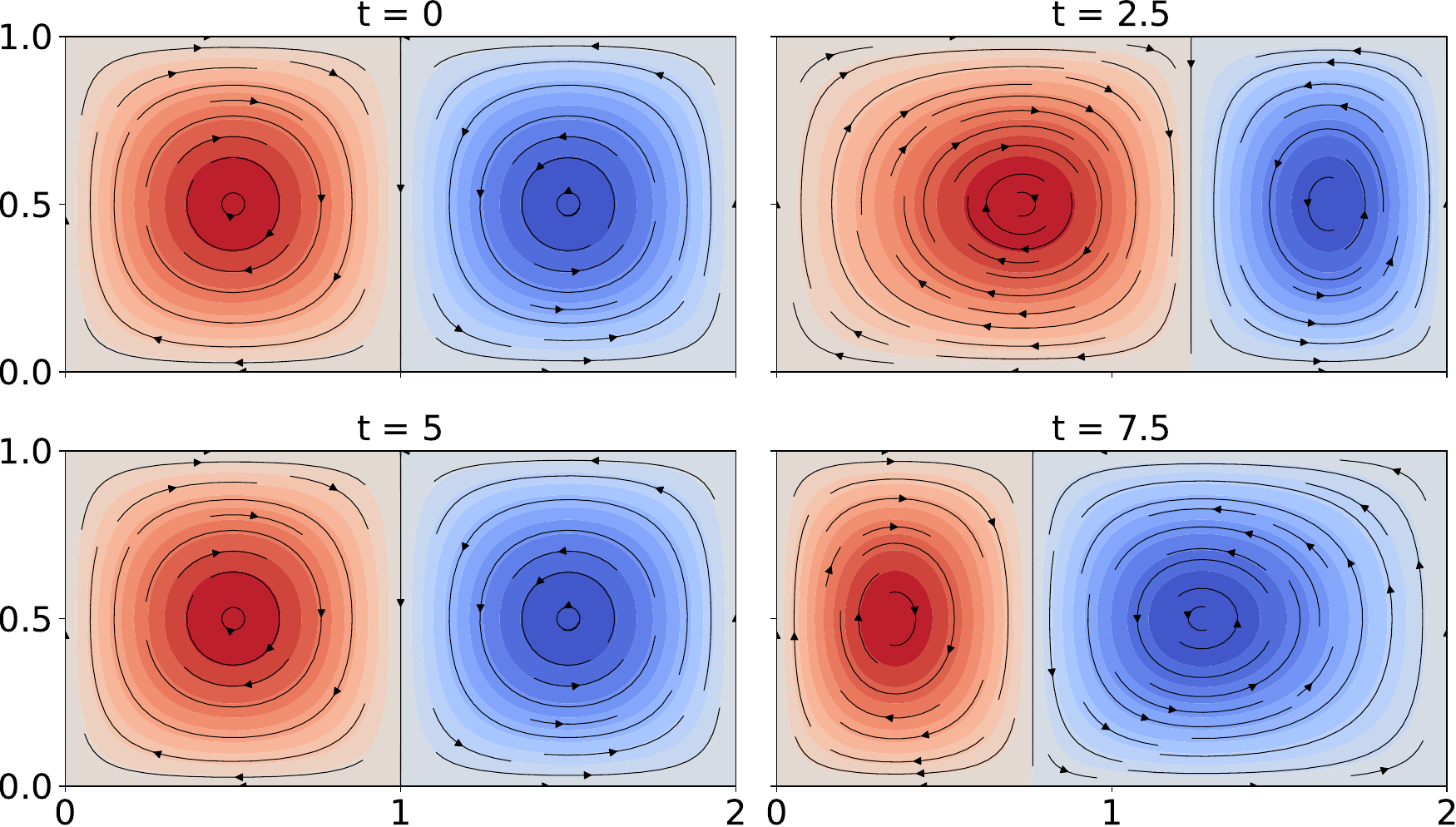}
    \caption{Stream lines of the time-dependent double gyre system at times $t=0$, $t=2.5$, $t=5$, and $t=7.5$. The boundary separating the gyres oscillates with amplitude $\epsilon = 0.25$ and period of oscillation $10$, so that at $t=0, 5, 10$ the boundary is located at $x=1$, and moves to $x=1.25$ at $t=2.5$ and to $x=0.75$ at $t=7.5$.}
    \label{fig:double_gyre_streamlines}
\end{figure}

In order to illustrate the performance of the spatio-temporal spectral clustering on real world data, we consider the time-dependent double gyre system, which can be regarded as a simplification of patterns that are often found in geophysical flows \cite{Coulliette_2000, Coulliette_2007}. This system is described by the equations
\begin{equation}
    \begin{split}
        \dot x &= - \pi A \sin \left(\pi f(x,t) \right) \cos\left(\pi y\right), \\
        \dot y &= \pi A \cos\left(\pi f(x, t )\right) \sin\left(\pi y \right) \frac{df}{dx},
    \end{split}
\end{equation}
with $ f(x,t) = \epsilon \sin(\omega t) x^2 + \left(1-2 \ts \epsilon\sin(\omega t) \right) x $, over the state space $ \mathbb{X} = [0, 2] \times [0, 1] $, see also~\cite{Forgoston_2011}. The system at times $ t = 0, 2.5, 5, 7.5 $ is shown in Figure \ref{fig:double_gyre_streamlines}. For this work, we follow Example 1 in \cite{Shadden_2005} and choose $ A = 0.1 $, $ \omega = \frac{2\pi}{10} $, and $\epsilon=0.25$. With these parameters, the system is time-dependent (as $\epsilon \neq 0$) and the boundary oscillates with period $10$, with an amplitude of oscillation of $0.25$ about the point $(1,0)$.

\begin{figure}
    \definecolor{matlab8}{RGB}{253, 231, 37}
    \definecolor{matlab9}{RGB}{68, 1, 84}
    \newcommand{\cdash}[1]{\textcolor{#1}{\rule[0.5ex]{1em}{0.3ex}}}
    \centering
    \begin{minipage}[b]{0.8\textwidth}
        \centering
        \subfiguretitle{(a)}
        \vspace{1ex}
        \includegraphics[width=\textwidth]{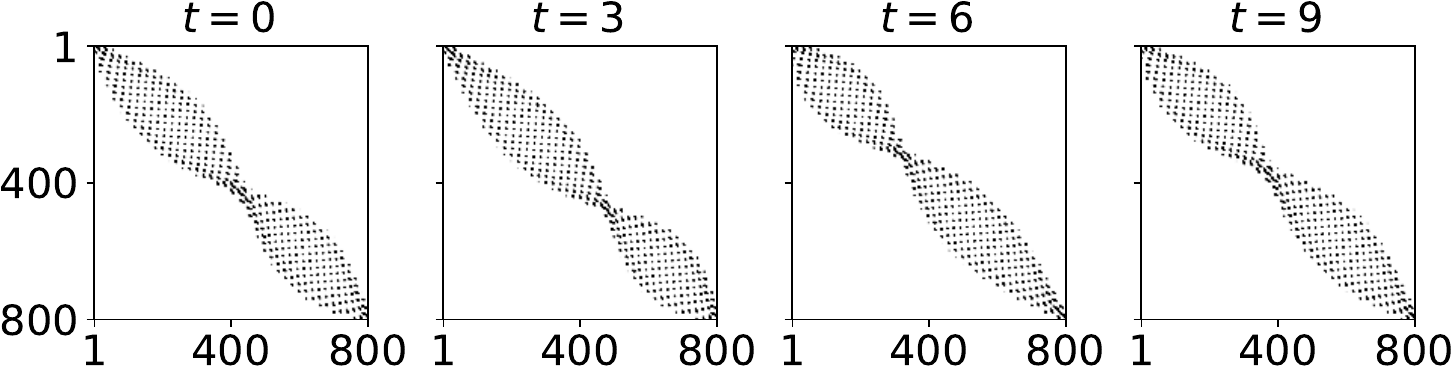}
    \end{minipage}
    \\[1.5ex]
    \begin{minipage}[t]{0.49\textwidth}
        \centering
        \subfiguretitle{(b)}
        \vspace{0.5ex}
        \includegraphics[height=3.7cm]{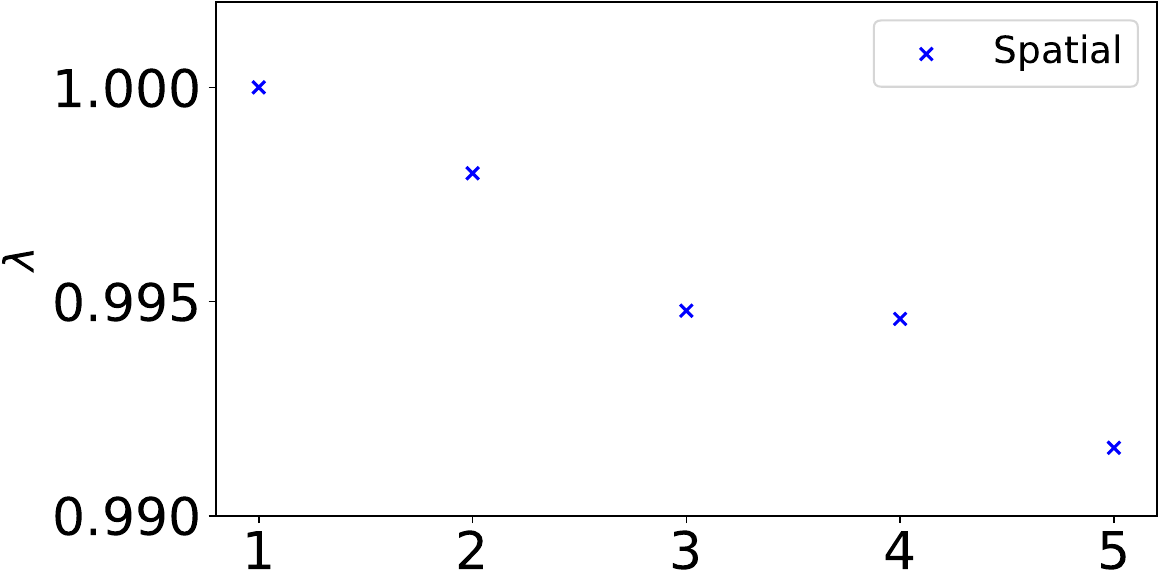}
    \end{minipage}
    \begin{minipage}[t]{0.49\textwidth}
        \centering
        \subfiguretitle{(c)}
        \includegraphics[height=3.73cm]{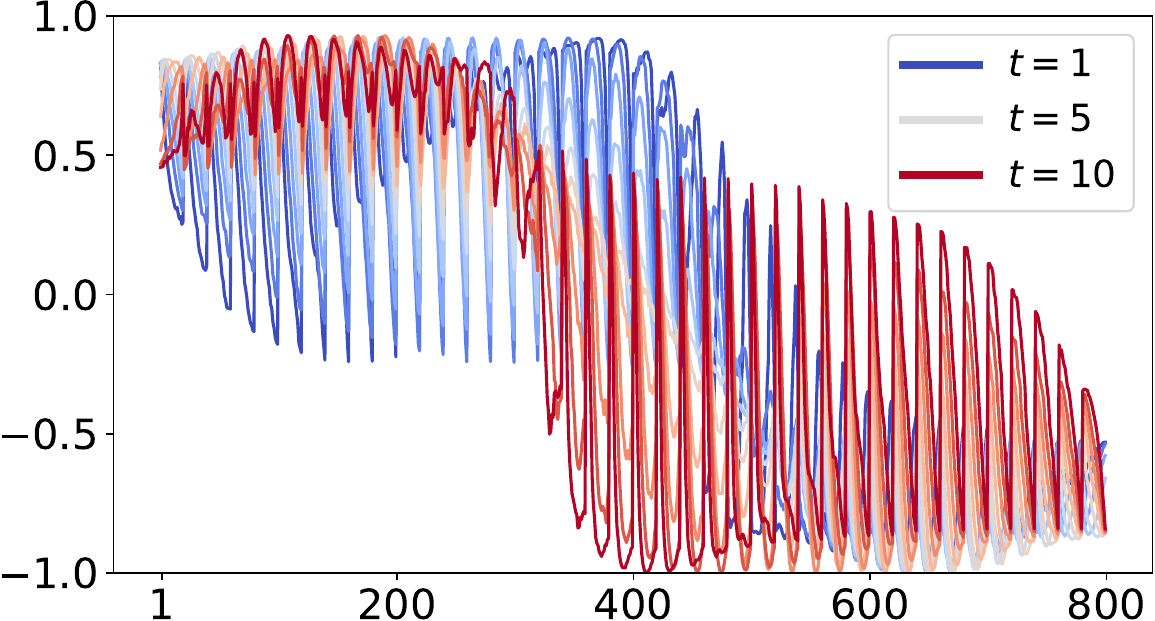}
    \end{minipage}
    \\[1.5ex]
    \begin{minipage}[b]{\textwidth}
        \centering
        \subfiguretitle{(d)}
        \vspace{1ex}
        \includegraphics[width=\textwidth]{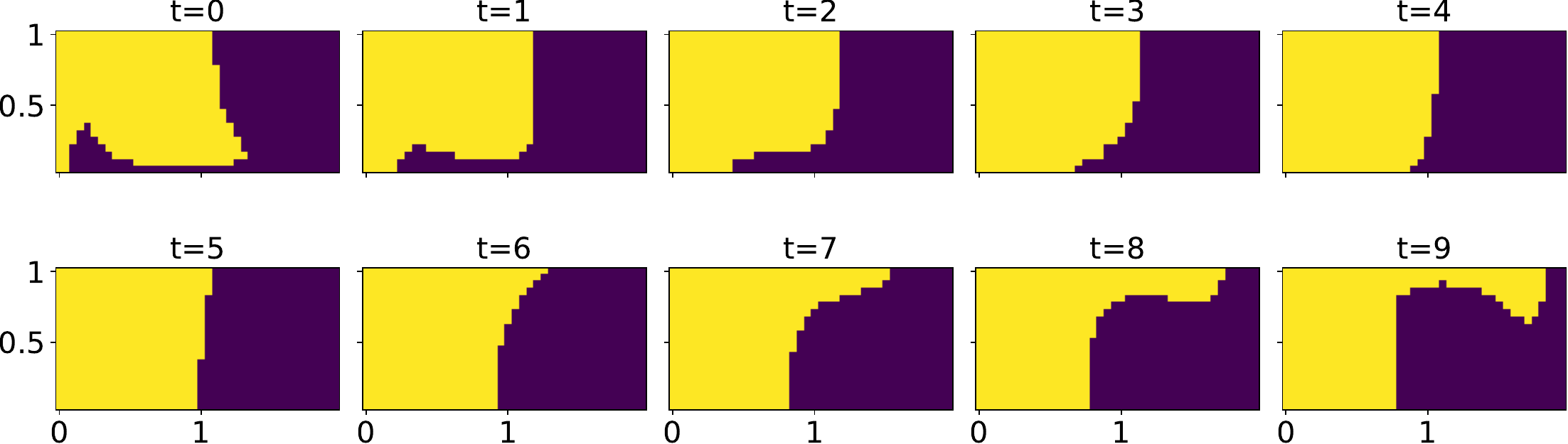}
    \end{minipage}
    \caption{(a) Transition matrices computed using transition probabilities estimated from time-dependent double gyre trajectory data, shown at times $t=0, 3, 6, 9$. (b) Eigenvalues of the matrix $\mathbf{C}$ for the time-evolving graph constructed using the transition matrices over the time interval $M=10$, with snapshots at $t=0, 1, \dots, 9$. A spectral gap between the second and third eigenvalues is observed. (c) The second eigenvector of $\mathbf{C}$, showing the changing size over time. (d) Cluster labels for each box in the discretized state space, where \cdash{matlab8} indicates cluster 1 (the left gyre) and \cdash{matlab9} indicates cluster 2 (the right gyre). The labels of each node are shown at each time view, and the expected dynamics of the double gyre system are observed in this labeling.}
    \label{fig:double_gyre_fig}
\end{figure}

We coarse-grain $\mathbb{X}$ into $40 \times 20$ boxes of equal size, and for each view $t=0, 1, \dots, 9$, we initialize 50 particles in each box and record the box index of each particle at time $t$ and at $t+1$. We can then define a transition matrix $S_t$ at time $t$ by counting the number of particles that begin in box $i$ and end in box $j$ in a matrix, and then dividing by the row sum to ensure that the matrix is row-stochastic. This approach is called \emph{Ulam's method}~\cite{Ulam60} and such box discretizations have, for instance, been used to compute almost invariant or metastable sets, see, e.g.,~\cite{DJ99, SH03}. Extensions to non-autonomous systems can also be found in \cite{FP14, fackeldey_2019}. We view the matrices $S_t$ as the transition matrices of a graph at a number of snapshots and consider the center of each box as a vertex. This results in a time-evolving graph on which we can apply the spatio-temporal spectral clustering. The transition matrices of this graph at times $t=0,3,6,9$ are shown in Figure~\ref{fig:double_gyre_fig}\ts(a).

In Figure~\ref{fig:double_gyre_fig}\ts(b) we can see the first five eigenvalues of the matrix $ \mathbf{C} $ for this time-evolving graph, and we note that in this case the first five are all spatial eigenvalues. That is, the dominant five eigenvalues and eigenvectors encode only spatial information. We also see that the moving boundary between the clusters is identified in the second eigenvector in Figure~\ref{fig:double_gyre_fig}\ts(c). This shows that the clustering algorithm can detect the growing and shrinking of the gyres as the system evolves. We also show, in Figure \ref{fig:double_gyre_fig}\ts(d), the cluster labels at each time step. Again, the expected behavior of the gyres is clearly seen in this visualization, indicating that the algorithm is capable of effectively clustering this simplified real-world flow.

\subsection{Summary of Results}

In both benchmark 1 and 2, we note that the presence of spatial and temporal eigenvalues does not allow us to choose $ k $ according to a spectral gap, and choosing the number of clusters for $ k $-means requires problem-specific knowledge or a trial--and--error approach. We use $ k = 3 $ for benchmark 1 and $ k = 4 $ for benchmark 2. The results of these benchmark comparisons are summarized in Table~\ref{tab:Benchmark graphs}. It is clear here that the spatio-temporal graph Laplacian approach performs very well on the benchmarks without any parameter tuning and is capable of clustering directed graphs effectively. Comparatively, spectral clustering using the supra-Laplacian can achieve very good results when tuned appropriately, but cannot detect off-diagonal directed clusters without a preprocessing step to facilitate this.

\begin{table}
    \caption{Comparison of Algorithm~\ref{alg:mCCA_sc} and the supra-Laplacian spectral clustering for benchmark graphs. For each method and graph, we compute the adjusted Rand index for the first and last time view. }
    \label{tab:Benchmark graphs}
    \scalebox{0.87}{
    \renewcommand{\arraystretch}{1.2}
    \newcolumntype{x}[1]{>{\centering\arraybackslash\hspace{0pt}}p{#1}}
    \begin{tabular}{|x{3cm}|x{2.8cm}|x{2.8cm}|x{2.8cm}|x{2.8cm}|x{2.8cm}|x{2.8cm}|}
        \hline
          & \multicolumn{2}{c|}{spatio-temporal graph Laplacian ARI}
          & \multicolumn{2}{c|}{supra-Laplacian ARI} \\ \hline
         & View 1 & View 10 & View 1 & View 10 \\ \hline
        benchmark 1 & 1.0 & 0.866 & 0.961 & 0.971 \\
        benchmark 2 & 1.0 & 1.0  & 0.749 & 0.565 \\ \hline
    \end{tabular}}
\end{table}

We also note that Section \ref{sec:double_gyre} illustrates the applicability of the spatio-temporal spectral clustering algorithm on data derived from real-world dynamical systems. We have shown that the algorithm can identify the oscillations observed in a time-dependent double gyre system, and there is no parameter tuning required to achieve these results.

\section{Conclusion} \label{sec:discussion}
We have proposed a novel framework for detecting and analyzing changing cluster structure in time-evolving graphs, and described connections to existing spectral clustering methods. Our approach extends canonical correlation analysis to maximize correlations across multiple time views, which leads to the definition of the spatio-temporal graph Laplacian. Via an analysis of the spectral properties of this Laplacian we can capture the temporal evolution of the graph. The results obtained from experiments on benchmark graphs and simplified real-world problems demonstrate the capabilities of the spatio-temporal graph Laplacian for detecting clusters that change over time, such as splitting, merging, and changing in size. The spatio-temporal graph Laplacian offers two significant advantages over other approaches, such as the supra-Laplacian. Firstly, the spatio-temporal graph Laplacian does not have any parameters which must be tuned, which can be computationally expensive and not straightforward. Secondly, the spatio-temporal graph Laplacian can naturally handle directed graphs without symmetrization and can detect off-diagonal clusters, expanding the range of applications for which it can be utilized.

However, there are still challenges and open problems to address. One issue is the selection of the number of eigenvalues and the number of clusters to choose in Algorithm~\ref{alg:mCCA_sc}. For typical spectral clustering algorithms, these numbers would be equal; we would cluster $ k $ eigenvectors to detect $ k $ clusters in a static graph. Determining the optimal value for each of these in time-evolving graphs is a well-known challenge, as we must also consider which eigenvalues correspond to spatial information and which correspond to temporal information. Further work is required to fully understand this, but first steps towards this for the supra-Laplacian can be seen in~\cite{gomez_2013}, where the asymptotic cases for the coupling strength are considered. A similar analysis for our method could yield key results for the spatio-temporal graph Laplacian method.

Another problem is a lack of well-established time-evolving benchmark graphs with ground truth clusterings, which means that we must use different metrics to evaluate the clustering algorithms on these graphs compared with static graphs. This issue is discussed in \cite{rossetti_2018}, and it is suggested that methods can be evaluated based on effectiveness when applied to real-world case studies, as investigated in Section \ref{sec:double_gyre}. An analysis of the performance of our algorithm on other datasets derived from fluid dynamics problems is a natural extension of this work, for example to identify merging or splitting vortices in fluid flows.

We also note that our formulation of mCCA considers only adjacent time views, but many other formulations have been proposed~\cite{cristianini_2004, kettenring_1971}. An exploration and comparison of these formulations could be of interest, and in particular we suggest that different formulations may yield better results for specific problems.

\section*{Acknowledgments}

M.\ T.\ was supported by the EPSRC Centre for Doctoral Training in Mathematical Modelling, Analysis and Computation (MAC-MIGS) funded by the UK Engineering and Physical Sciences Research Council (grant EP/S023291/1), Heriot--Watt University and the University of Edinburgh. N.\ Dj.\ C. is partially funded by the Deutsche Forschungsgemeinschaft (DFG) under Germany’s Excellence Strategy through grant EXC-2046 The Berlin Mathematics Research Center MATH+ (project no.\ 390685689).

\section*{Author Declarations and Data Availability}

The authors have no conflicts to disclose. The data that support the findings of this study are available from the corresponding author upon request.

\bibliographystyle{unsrturl}
\bibliography{SGL}

\end{document}